\documentclass[reqno, 11pt]{amsart}
 \usepackage{color,mathrsfs}
 \usepackage{appendix}
\usepackage[utf8]{inputenc}
\usepackage[colorlinks=true, pdfstartview=FitV, linkcolor=blue, citecolor=blue, urlcolor=blue,pagebackref=false,hyperindex,breaklinks]{hyperref}
\usepackage{calc,amsfonts,amsthm,amscd,epsfig,psfrag,amsmath,amssymb,enumerate,graphicx}
\usepackage[showlabels,sections,floats,textmath,displaymath]{}

\reversemarginpar
\newlength\fullwidth
\numberwithin{equation}{section}

\DeclareMathSymbol{\leqslant}{\mathalpha}{AMSa}{"36} 
\DeclareMathSymbol{\geqslant}{\mathalpha}{AMSa}{"3E} 
\DeclareMathSymbol{\eset}{\mathalpha}{AMSb}{"3F}     
\renewcommand{\leq}{\;\leqslant\;}                   
\renewcommand{\geq}{\;\geqslant\;}                   
 \def\1{\ifmmode {1\hskip -3pt
    \rm{I}} \else {\hbox {$1\hskip -3pt \rm{I}$}}\fi}

\newcommand{\gmv}{|g|_{\mathrm{max}}}
\newcommand{\gm}{g_{\mathrm{max}}}

\newcommand{\Z}{\mathbb{Z}}

\renewcommand{\l}{\lambda}

\renewcommand{\d}{\delta}

\newtheorem{theorem}{Theorem}[section] 
 
\newtheorem{lemma}[theorem]{Lemma} 
\newtheorem{proposition}[theorem]{Proposition}

\newtheorem{remark}[theorem]{Remark}

\newcommand{\go}{\omega}
\newcommand{\N}{\mathbb N}

\newcommand{\cA}{\ensuremath{\mathcal A}} 
\newcommand{\cB}{\ensuremath{\mathcal B}} 
 
\newcommand{\cD}{\ensuremath{\mathcal D}} 
 
\newcommand{\cF}{\ensuremath{\mathcal F}} 
\newcommand{\cG}{\ensuremath{\mathcal G}} 
\newcommand{\cH}{\ensuremath{\mathcal H}}

\newcommand{\cP}{\ensuremath{\mathcal P}}

\newcommand{\cS}{\ensuremath{\mathcal S}}

\newcommand{\am}{a_{\mathrm{min}}} 
\newcommand{\bB}{{\ensuremath{\mathbf B}} }

\newcommand{\bbR}{{\ensuremath{\mathbb R}} }

\newcommand{\bbZ}{{\ensuremath{\mathbb Z}} }

\newcommand{\gep}{\varepsilon}

\newcommand{\ind}{{\bf 1}}
\newcommand{\gd}{\delta}
\newcommand{\dd}{\mathrm{d}}

\newcommand{\gl}{\lambda}

\newcommand{\DD}{\mathcal{D}}

\newcommand{\be}{\mathbf{e}}

\newcommand{\mal}{\mathcal A_L}

\begin{document}
\begin{abstract}

  Let $\DD$ be a bounded, smooth enough domain of $\bbR^2$. For $L>0$
  consider the continuous time, zero-temperature heat bath dynamics
  for the nearest-neighbor Ising model on $(\mathbb Z/L)^2$ (the
  square lattice with lattice spacing $1/L$) with initial condition
  such that $\sigma_x=-1$ if $x\in \DD$ and $\sigma_x=+1$ otherwise.
  We prove the following classical conjecture
  \cite{cf:Spohn,cf:CL} due to H. Spohn: In the diffusive limit where time is rescaled
  by $L^2$ and $L\to\infty$, the boundary of the droplet of "$-$"
  spins follows a \emph{deterministic} anisotropic curve-shortening
  flow, such that the normal velocity is given by the local curvature
  times an explicit function of the local slope. Locally, in a
  suitable reference frame, the evolution of the droplet boundary
  follows the one-dimensional heat equation.

  To our knowledge, this is the first proof of mean curvature-type
  droplet shrinking for a lattice model with genuine microscopic
  dynamics.

  An important ingredient is our recent work \cite{LST}, where the
  case of convex $\DD$ was solved. The other crucial point in the
  proof is obtaining precise regularity estimates on the deterministic
  curve shortening flow. This builds on geometric and analytic ideas
  of Grayson \cite{Grayson}, Gage-Hamilton \cite{GH}, Gage-Li
  \cite{GL1,GL2}, Chou-Zhu \cite{ChuZu} and others.\end{abstract}

\title[The heat equation shrinks Ising droplets to points]{
The heat equation shrinks Ising droplets to points}

\author[H. Lacoin]{Hubert Lacoin}
\address{H. Lacoin, 
CEREMADE - UMR CNRS 7534 - Universit\'e Paris Dauphine,
Place du Mar\'echal de Lattre de Tassigny, 75775 CEDEX-16 Paris, France. \newline
e--mail: {\tt lacoin@ceremade.dauphine.fr}}

\author[F. Simenhaus]{Fran\c{c}ois Simenhaus}
\address{F. Simenhaus, 
CEREMADE - Universit\'e Paris Dauphine - UMR  CNRS 7534,
Place du Mar\'echal de Lattre de Tassigny, 75775 CEDEX-16 Paris France. \newline
e--mail: {\tt simenhaus@ceremade.dauphine.fr}}
\author[F. L. Toninelli]{Fabio Lucio Toninelli} 
\address{F. L. Toninelli, Universit\'e de Lyon and CNRS, Universit\'e
  Lyon 1, Institut Camille Jordan, 43 bd du 11 novembre 1918, 69622 Villeurbanne, France.  \newline e--mail: {\tt
    toninelli@math.univ-lyon1.fr}}

    \maketitle

  \section{Introduction}
  A basic problem in non-equilibrium statistical mechanics is the
  following \cite{cf:Spohn} : Take a microscopic statistical mechanics
  model at sufficiently low temperature so that there are, say, two
  pure thermodynamics phases.  Assume that such system evolves
  according to a microscopic dynamics defined via local evolution
  rules. Then, the goal is to derive macroscopic, deterministic
  equations 
  which describes, on large space-time scales, the evolution of
  spatial boundaries separating the two coexisting thermodynamic
  phases.  An example to keep in mind is the nearest-neighbor Ising
  model on $\bbZ^d$, $d\ge2$.  Below the critical temperature, 
  in the absence of an external magnetic field, there are two translation invariant
  equilibrium Gibbs measures (the ``$+$'' and the ``$-$''
  thermodynamic phase).
  One can then easily define a Markov dynamics,
   the so-called Glauber dynamics,
  where individual spins are flipped with rates chosen so that the
  Gibbs measures are invariant and reversible. If at time zero a
  region of the space is occupied by the ``$+$'' phase and the rest by
  the ``$-$'' phase, we are interested in how the shape of these
  region will evolve with time.

  If the dynamics does not conserve the order parameter (e.g. the
  total magnetization for the Ising model), it is well understood
  phenomenologically \cite{cf:Lifshitz} that a droplet of one phase
  immersed in the opposite phase will shrink in order to decrease its
  surface tension until it disappears in finite time; also (roughly
  speaking) the normal velocity at a point of its boundary will be
  proportional to the local mean curvature.  Based on this idea, one
  expects (``Lifschitz law'') that, if the initial droplet has diameter $L$, it will
  ``evaporate'' within a time of order $L^2$ (as would be the case for
  a sphere evolving via mean curvature motion). Moreover, the droplet
  evolution should become deterministic and follow some version of a
  mean curvature flow in the ``diffusive limit'' where $L\to\infty$,
  space is rescaled by $L$ (the initial droplet is then of size
  $O(1)$) and time is accelerated by $L^2$. The resulting large-scale
  deterministic evolution equation will in general be anisotropic
  (i.e. the normal velocity will depend also  on the local orientation of
  the droplet boundary) when the microscopic model is defined on a
  lattice, as it is for the Ising model.

  The main difficulty in implementing this program is that there is no
  obvious way how to separate the ``fast modes'' related to relaxation
  inside the bulk of the pure phases from the ``slow modes'',
  responsible for the $L^2$ time scaling, related to the interface
  motion. Such problem is absent in so-called ``effective interface
  models'' of Ginzburg-Landau $\nabla\phi$ type: for these models,
  under an assumption of strict convexity of the interaction, Funaki
  and Spohn \cite{FS} derived the full mean-curvature motion in the
  diffusive scaling. Another case \cite{K1,K2,K3} where mean-curvature motion 
 is known to appear in the scaling limit are spin  models with
Kac-type interactions (the interaction range tends to infinity with the droplet size):  in this case, however, the system is  close to
mean-field, the deterministic flow is isotropic and  there is no sharp
interface separating the phases.  


 Results are much more incomplete  for genuine lattice models: for
instance, for the two-dimensional  Ising model at low but non-zero  temperature $T$, 
it is only known that a droplet of ``$-$ phase'' immersed in the ``$+$ phase'' will disappear in a time 
of  order at most $L^{c(T)\log L}$ \cite{cf:LSMT}, to be compared with the 
expected $L^2$ scaling. 

In the present work, we study the two-dimensional nearest-neighbor Ising
model on the square lattice, at zero temperature. Each spin variable
$\sigma_x=\pm1$ is updated on average once per 
time unit: after the
update, the spin takes the same sign as the majority of its $4$
neighbors, or the value $\pm1 $ with equal probabilities in case of a
tie.  Assume that the initial ``$-$'' droplet, when the lattice spacing tends to zero,
converges to a smooth enough domain $\DD$ of $\bbR^2$. Then, in the
diffusive scaling limit the droplet boundary is expected to be given
by a deterministic evolving curve $\gamma(t)$. Such curve should move
according to the following ``(anisotropic) curve-shortening flow'':
the normal velocity equals the local (signed) curvature, times a
function $a(\theta)$ with $\theta$ the local tangent. The function
$a(\cdot)$ is explicitly given, cf. \eqref{eq:a}.

This result was conjectured in 
\cite{cf:Spohn}  by Spohn, who gave  some very reasonable supporting arguments, based
on the local analysis of the dropled boundary evolution in terms of
interacting particle systems. 
In \cite{cf:CL}, Cerf and Louhichi computed the
``drift at time $0$'' of the droplet (for the non-modified dynamics),
but their result does not allow to get any information on the evolution
for positive time $t>0$. The full convergence to the curve shortening flow for initial \emph{convex} droplets $\DD$ was recently obtained in \cite{LST}.
 The main result of the present work, Theorem
\ref{mainres},  is a proof of Spohn's conjecture, for smooth enough
initial droplets $\DD$, without any convexity assumption.
Smoothness of the initial condition is required essentially so that the 
limit flow is unambiguously defined.


As it was the case also in \cite{LST} for the convex initial condition, a preliminary
but essential step before proving convergence of the stochastic
evolution to the limit deterministic one is to show that the
anisotropic curve-shortening flow does admit a global (in time)
solution, and that it does not develop singularities before it shrinks
to a point. For the isotropic case ($a(\cdot)\equiv
1$) this was proven by Grayson in a celebrated work \cite{Grayson}.
In the anisotropic case, Grayson's result has been extended
(e.g. \cite{cf:Oaks,CZlibro}) under the assumption that $a(\cdot)$ is
at least $C^2$. The reason is very simple: the first step in the
procedure is to write down the evolution equation satisfied by the curvature of
$\gamma(t)$, and in such equation a second derivative of $a(\cdot)$
appears.  In our case, the anisotropy function $a(\cdot)$ is not even
$C^1$ (which reflects singularities of the zero-temperature surface
tension of the Ising model). To prove existence, uniqueness and
regularity of the (classical) solutions of the curve-shortening flow (cf.\ Theorem \ref{th:existence}), we
first regularize the function $a(\cdot)$ and then analyze the
regularized equation following the ideas of
\cite{Grayson,ChuZu,CZlibro,cf:Oaks}.  Of course, it is crucial to
check that all the estimates we need are uniform in the regularization
parameter, which is sent to zero in the end. Let us emphasize that the
regularity estimates of \cite{CZlibro,ChuZu,cf:Oaks} are far from
being quantitative in terms of the smoothness of the anisotropy function
$a(\cdot)$.

\smallskip

Comparing our present result with that of \cite{LST}, it is important
to realize that dropping the convexity assumption is not at all a
technical point.  First of all, various monotonicity arguments that
were used in \cite{LST} do not work here. The basic reason is that
such ideas crucially 
relied on the fact that, in the convex case, the
normal velocity is always directed inward (which is clearly false for
non-convex droplets, at points where the curvature is negative).
Secondly, proving existence and regularity of solution requires very
different analytic and geometric arguments in the non-convex case with respect to the convex
one (there, we were able to use ideas from \cite{GH,GL1,GL2}).  At any
rate, our previous result \cite{LST} is important in Section
\ref{sec:mainres}, where the evolution of the droplet boundary is controlled
by locally comparing it with that of a suitable convex
droplet.

\medskip

Let us mention  some recent related works by one of the authors.
In \cite{cf:L2} 
the issue of the evolution of a convex planar
``minus droplet'' in the presence of a positive magnetic field has been
investigated.  In this case the right time scaling is $L$ instead of
$L^2$ for a droplet of size $O(L)$ and the scaling limit is given by
the anisotropic eikonal equation: the drift of the interface loses its
dependence on the curvature. 
Generalizing such a result in higher dimensions is a very challenging problem, see e.g. \cite{phys} for a non-rigorous attempt in this direction.  In \cite{cf:L1},
the dynamical evolution of a half-droplet on a substrate that
attracts the interface (a situation that also corresponds to  ``dynamical polymer pinning'') was studied, and
the scaling limit was shown to be the solution of Stefan-type equation
where the motion of the point of contact between the droplet and the
substrate depends on the local curvature.

\smallskip

We close this introduction by mentioning a couple of intriguing
open problems. First of all, one would like to know what are the finite-$L$ 
fluctuations of the droplet boundary around its limit shape
$\gamma(t)$, along the evolution.  Secondly, it is natural to wonder
what happens for the zero-temperature dynamics of the \emph{three- (or
  higher-) dimensional} Ising model.  Recently, a weak version of the
Lifshitz law was proven for the three-dimensional Ising model at zero
temperature: the evaporation  time of a ``$-$'' droplet is of order
$L^2$, up to multiplicative logarithmic corrections
\cite{cf:CMST}.  An analogous upper bound was proven in higher
dimensions \cite{cf:L}.  However, it is still not clear (even at a heuristic level) 
what should be the precise macroscopic equation, analogous to \eqref{eq:mc}, describing the 
droplet evolution in the diffusive limit.

\medskip

\section{Model and results}

Given $L\in \N$ we
consider the zero-temperature stochastic Ising model on $(\mathbb
Z/L)^2$ 
(the square lattice with lattice spacing $1/L$). The state space is
the set $\Omega=\{-1,+1\}^{(\bbZ/L)^2}$ of spin configurations
$\sigma=(\sigma_x)_{x\in (\bbZ/L)^2}$ with $\sigma_x=\pm1$.
The dynamics is a Markov process $(\sigma(t))_{t\ge0}$, with
$\sigma(t)=(\sigma_x(t))_{x\in(\bbZ/L)^2}\in\Omega$.
Each spin $\sigma_x$ is updated with unit rate: when the update
occurs, $\sigma_x$ takes the value of the majority of its four neighbors, or
takes values $\pm1$ with equal probabilities if exactly two neighbors
are $+1$ and two neighbors are $-1$.
\medskip

We consider a compact, simply connected subset $\mathcal D\subset
[-1,1]^2$ whose boundary $\partial \mathcal D$ is a Jordan curve of
finite length. The initial condition of the stochastic dynamics will
be set to be 
``$-$'' inside $\mathcal D$ and ``$+$'' outside:
\begin{equation}\label{bien}
\sigma_x(0)=\begin{cases} - 1 \quad &\text{ if } x\in (\bbZ/L)^2\cap \mathcal D,
                          \\ +1 \quad &\text{ otherwise}.
            \end{cases}
 \end{equation}
We want  to compute the scaling limit  of the set of ``$-$'' spins at
positive times, when $L\to \infty$.
In order to identify a set of ``$-$'' spins as a subset of $\bbR^2$,
let for $x\in (\Z/L)^2$
\begin{equation}
\mathcal C_x:= x+[-1/(2L),1/(2L)]^2
\end{equation}
be the square of side $1/L$ centered at $x$ and define 
\begin{equation}\label{mal}
\mal(t):=\bigcup_{\{y: \ \sigma_y(t)=-1\}} \mathcal C_y,
\end{equation}
which is the ``$-$ droplet'' at time $t$ for the dynamics.  

 
\medskip

Our goal is to prove that, as $L\to\infty$, $\mal(L^2 t)$  converges to the  compact set $\cD_t$ whose boundary  
$\gamma(t)=\partial\mathcal D_t$ is the solution of the anisotropic
curve shortening flow 
\begin{equation}\label{eq:mc}
 \partial_t \gamma= a(\theta) k {\bf N}
\end{equation}
with initial condition $\gamma(0):= \partial \mathcal D$. This equation has to be read as follows.
The normal velocity at a point
$p\in\gamma(t)$ is given by the curvature $k$ at point $p$ times $a(\theta(p))$,
with 
\begin{equation}
\label{eq:a}
a(\theta)=\frac{1}{2(|\cos(\theta)|+|\sin(\theta)|)^2}, \qquad 0\leq \theta\leq 2\pi
\end{equation}
and $\theta(p)$ the  tangent angle to $\gamma(t)$ at $p$. The
normal vector ${\bf N}$ at point $p$ points inward and the curvature is positive
(resp. negative) at
points of local convexity (resp. concavity) of $\gamma(t)$.

\medskip
Since $a(\theta)$ is not differentiable 
for $\theta$ multiple of $\pi/2$, the existence of a solution for
\eqref{eq:mc} does not follow from the standard literature, that
assumes $a(\cdot)$   to be at least $C^2$ (see \cite{cf:Oaks}). Our first result is an
existence, uniqueness and regularity theorem for the solution of \eqref{eq:mc}. Define
\begin{eqnarray}
  \label{eq:12}
T=T(\cD)=\frac{Area(\mathcal D)}{\int_0^{2\pi}a(\theta) d\theta}
=\frac{Area(\mathcal D)}2.
\end{eqnarray}

\begin{theorem}
  \label{th:existence}
Consider a domain $\cD\subset[-1,1]^2$ whose boundary is a Jordan curve $\gamma(0)$ of finite
length, with curvature everywhere defined and $C^\infty$ 
as a function of the arc-length coordinate. Suppose moreover that
$\gamma(0)$ has a finite number of inflection points.\\
There exists a unique solution $(\gamma(t))_{t\le T(\cD)}$ of
\eqref{eq:mc} that is a Jordan curve for $t< T$ and:
\begin{enumerate}
\item  The area enclosed by $\gamma(t)$ is $Area(\cD)-2t=2(T-t)$ and
  $\gamma(t)$ shrinks to a point ${X}$ when $t\to T$.
\item For every $s<T$,  the curvature function is equicontinuous on $[0,s]$ in the
following sense: for every $\epsilon>0$ there exists
$\delta=\delta(\epsilon,s,\gamma(0))>0$ such that if $t,t'\le s$ and $p\in\gamma(t),p'\in\gamma(t')$ with $|t-t'|\le \delta, |p-p'|\le \delta$
then $|k(p,t)-k(p',t')|\le \epsilon$. 
\end{enumerate}
\end{theorem}
Our main result gives convergence of the stochastic droplet
$\mathcal A_L(L^2 t)$ to
the deterministic flow $\mathcal D_t$, that is the compact domain enclosed by
$\gamma(t)$ (for $t>T$, we set by convention $\mathcal D_t:=\{X\}$).

We introduce some notations.
For $\eta>0$ let $\cB(x,\eta)$ denote the ball of radius $\eta$
centered at $x\in\mathbb R^2$
and 
for any compact set $\mathcal C\subset \bbR^2$, we define
\begin{equation}
\label{eq:cdelta}
 \mathcal C^{(\eta)}:=\bigcup_{x\in \mathcal C} \cB(x,\eta),\quad\quad
 \mathcal C^{(-\eta)}:=\left(\bigcup_{x\notin \mathcal C} \cB(x,\eta)\right)^c.
\end{equation}
Finally, we will say that an event holds with high probability
(w.h.p.) if its probability tends to $1$ as $L$ tends to infinity.

%
%

\begin{theorem}\label{mainres}
Consider $\cD$ such that $\gamma(0)=\partial \cD$ satisfies the assumptions of Theorem \ref{th:existence}.
Let us consider the zero temperature stochastic Ising model with
initial condition \eqref{bien}.
Then for any $\eta>0$ the following holds w.h.p.: 
\begin{enumerate}
\item for all $t\geq 0$,
\begin{equation}\label{trezi}
  \mathcal D^{(-\eta)}_t\subset \mal( L^2 t)\subset  \mathcal D^{(\eta)}_t;
\end{equation}
\item for all $t\ge T+\eta$, $\mal( L^2 t)$ is empty.
\end{enumerate}
\end{theorem}

%
%
We emphasize that the regularity estimates stated in Theorem
\ref{th:existence} are not given just for the sake of
completeness, but on the contrary are crucial in the proof of Theorem \ref{mainres}.

\subsection{Generalizations and open problems}

Let us mention a few immediate generalizations of our result, and an interesting
open problem.

\smallskip
\begin{enumerate}[i.]
\item \emph{More general initial condition}. Instead of \eqref{bien}, let us
assume only that the (possibly random) initial  droplet $\mal(0)$ converges w.h.p. in
Hausdorff distance to $\cD$ 
as $L\to\infty$. Then, \eqref{trezi} still holds.
Just note that, for any given $\epsilon>0$,
w.h.p. $\cD^{(-\epsilon)}\subset \mal(0)\subset
\cD^{(\epsilon)}$: then the claim follows from Theorem  \ref{mainres}
plus monotonicity of the dynamics, cf. Section \ref{monoton}.

\item \emph{Non-connected initial droplet}. If $\cD$ is not connected
  but each of its connected components $\cD_i, i=1,\dots,k$ verifies the assumptions of
  Theorem \ref{th:esistenza}, it is easy to see that Theorem
  \ref{mainres} applies to each  of the $k$ components of the
  stochastic droplet $\mal(t)$. Essentially, the various components
  evolve independently.

\item \emph{Non-simply connected initial droplet}. Suppose that $\cD$
  is compact, connected but non-simply connected (say, an
  annulus). If each connected component $\gamma_i(0),i=1,\dots,k$ of $\gamma(0)=\partial \cD$
  verifies the assumptions of Theorem \ref{th:esistenza}, define
  $\cD_t$ to be the domain with boundary $\partial \cD_t=\cup_i
  \gamma_i(t)$. Then, it is not hard to see that Theorem \ref{mainres}
  still holds. Again, roughly speaking the $k$ macroscopic components of the
  boundary of $\mal(L^2t)$ evolve essentially independently and
  approach the deterministic evolution of the $k$ components of $\partial \cD_t$.

\item \emph{$\partial\cD$ is not a simple curve}. If $\partial \cD$ has
self-intersections the situation is definitely
more subtle. To fix ideas, consider the case of Figure \ref{fig:8_1}.
Note that $\cD$ can be seen either as the
$\gep\to0$ limit of a domain $\cF^{\gep}$ with Jordan
boundary and an $\gep$-narrow pinch or as the limit of two $\gep$-close
simple domains $\cG^{\gep},\cH^{\gep}$.
 In this case, we expect that the evolution of $\mal(L^2
t)$ remains random in the $L\to\infty$ limit.
More precisely, we expect that the Ising droplet will follow with
some
probability $p$ the \emph{deterministic} evolution $\lim_{\gep\to0}
(\cF^\gep)_t$ and with probability $1-p$ the deterministic evolution $\lim_{\gep\to0}
[(\cG^\gep)_t\cup (\cH^\gep)_t]$. Here, $(\cF^\gep)_t$ is
the domain enclosed by the solution of \eqref{eq:mc} with initial
condition $\partial \cF^\gep$, and similarly for $(\cG^\gep)_t,(\cH^\gep)_t$.
More importantly, we expect that the law of the limit evolution
(i.e. the probability $p$)
depends crucially on the way the initial droplet $\mal(0)$
microscopically approximates $\cD$.
This will be considered in future work.
\begin{figure}[hlt]
 \begin{center}
 \leavevmode 
 \epsfxsize =10  cm
 \psfragscanon
 \psfrag{eps}{$\gep$}
\psfrag{eps2}{$\gep$}
\psfrag{D}{$\mathcal{D}$}
\psfrag{F}{$\mathcal{F}^\gep$}
 \psfrag{G}{$\mathcal{G}^\gep$}
 \psfrag{H}{$\mathcal{H}^\gep$}
 \epsfbox{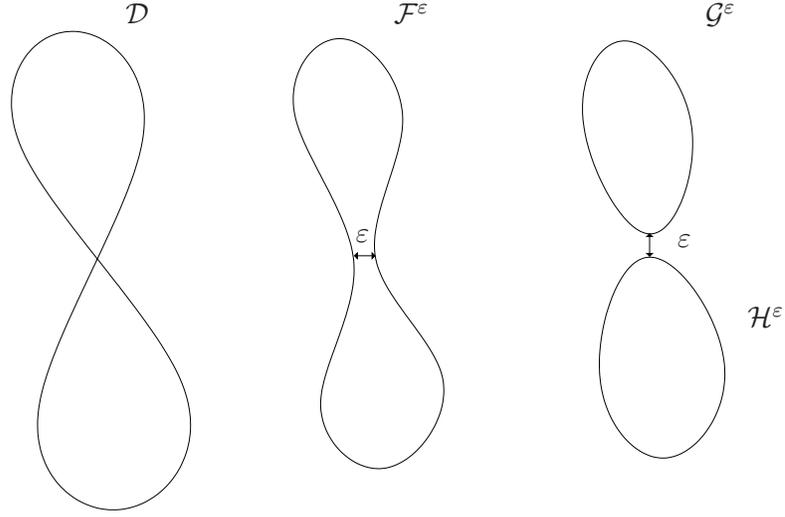}
 \end{center}
 \caption{An initial droplet whose boundary is not a simple curve.}
 \label{fig:8_1}
 \end{figure}

\begin{figure}[hlt]
 \begin{center}
 \leavevmode 
 \epsfxsize =10  cm
 \psfragscanon
 \psfrag{Q1}{\small $Q_1$}
\psfrag{Q2}{\small $Q_2$}
\psfrag{Q3}{\small $Q_3$}
\psfrag{Q4}{\small $Q_4$}
 \psfrag{C1}{\small $C^1(t)$}
  \psfrag{C2}{\small $C^2(t)$}
 \psfrag{Cgo}{\small $C^\go(t)$}
 \psfrag{Cgo1}{\small $C^\go_1(t)$}
  \psfrag{Cgo2}{\small $C^\go_2(t)$}
 \epsfbox{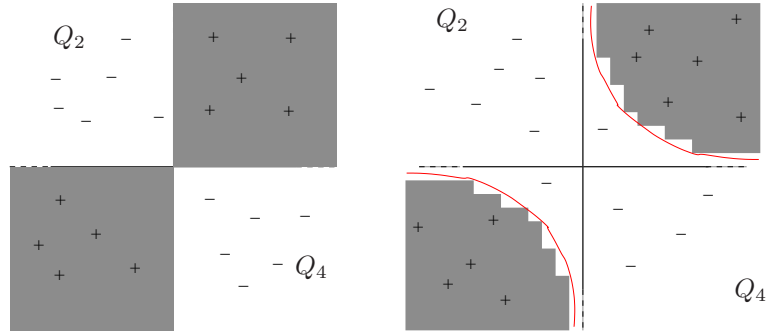}
 \end{center}
 \caption{ Assume that initially spins are ``$+$'' in the
first and third quadrants $Q_1,Q_3$ of $(\bbZ/L)^2$ and ``$-$'' in the second and fourth
quadrants $Q_2,Q_4$ (left drawing). Then, by symmetry, with
probability $1/2$ the boundary of $\mal(L^2t)$ converges to $\partial
(Q_1)_t\cup \partial
(Q_3)_t $ (right drawing) and with probability $1/2$ to $\partial
(Q_2)_t\cup \partial
(Q_4)_t $.}
 \label{fig:quadranti}
 \end{figure}
\end{enumerate}

  \section{Proof of Theorem \ref{th:existence}}
\label{sec:provaesistenza}


\subsection{ A few general facts on (anisotropic) curve shortening flows}
Let us pretend for a moment that equation 
\eqref{eq:mc} is replaced by 
 \begin{equation}\label{e:mceq}
\partial_t\gamma(t)=A(\theta) k{\bf N}
 \end{equation} 
where $A(\cdot)$ is a smooth (at least $C^2$) function, that is strictly positive and $\pi$-periodic. 
Then, it is known from
\cite{cf:Angenent1}
that \eqref{e:mceq} does admit a solution $\gamma(t)$ that is a Jordan
curve. Furthermore from  \cite{cf:Oaks} the 
solution exists until a maximal time where it shrinks to a point.
The function
$a(\cdot)$ in \eqref{eq:a} is instead only Lipschitz, which is why we
cannot apply known results to our case.

Call $s$ the arc-length coordinate on $\gamma(t)$. It is not important
to establish what point of $\gamma(t)$ is assigned the coordinate
$s=0$, since we will always evaluate only derivatives or differences with respect to
$s$.  In the following equations (\eqref{eq:7} to \eqref{eq:dtheta}), $\partial_s$ denotes derivation
w.r.t. arc-length at fixed time, while $\partial_t$ means derivation
along the flux lines described by points $p(t)\in\gamma(t)$ who move
with velocity at time $t$ given by $k(p(t))A(\theta(p(t)))$ times the
normal vector ${\bf N}$
to $\gamma(t)$ at point $p(t)$.  It is well known 
(see \cite[Lemma $1.4$]{Grayson} in the isotropic case) 
that the
derivatives with respect to $t$ and $s$ do not commute, as motion
affects the arc-length: in fact, we have
\begin{eqnarray}
  \label{eq:7}
  \partial_t\partial_s=\partial_s\partial_t+A k^2.
\end{eqnarray}

 We collect here a few useful formulas. First of all,
define for $p\in\gamma(t)$
\begin{eqnarray}
  \label{eq:g}
  g(p)=A(\theta(p))k(p)
\end{eqnarray}
with $\theta(p)$ the tangent angle at the point $p$.
Then, one has (see for instance \cite[Ch. 1]{CZlibro}; for the
isotropic case $A\equiv 1$, see \cite{GH})
\begin{eqnarray}
\label{eq:dk}
\partial_t k&=&\partial^2_sg+A k^3\\
  \label{eq:dg}
  \partial_t g&=&\partial_s(A\partial_s g)+A^2 k^3\\
  \label{eq:dtheta}
  \partial_t \theta&=&\partial_s(Ak)\\
  \label{eq:dleng}
 \partial_t {\mathcal L}&=&-\int_{\gamma}A k^2 ds
\\
  \label{eq:darea}
 \partial_t{\mathcal A}&=&-\int_0^{2\pi}A(\theta)d\theta
\end{eqnarray}
with $\mathcal L(\gamma), \mathcal A(\gamma) $  the length of $\gamma$ and
the area enclosed by it.
Remark also that $
  \partial_s\theta=k
$
(this is a simple geometric fact that has nothing to do with the curve shortening flow).

\subsection{The existence theorem}

Let us call $(a^{\go}(\cdot))_{\go\in(0,1)}$ a family of $C^\infty$
regularizations of $a(\cdot)$ that are uniformly Lipschitz  (this is possible because
$a(\cdot)$ itself is $1$-Lipschitz), have the same $\pi/2$-periodicity
as $a(\cdot)$ and converge
uniformly to $a(\cdot)$ when $\go\to 0$.  
From the previous section, we know that
\eqref{eq:mc} has a solution if $a$ is replaced by $a^\go$.
Let $\gamma^\go(t)$ be this solution and $k^\go=(k^\go(p,t))_{t\ge
  0,p\in \gamma^\go(t)}$ denote its curvature function.
The time when the curve shrinks to a point is given by (cf.\ \eqref{eq:darea})
\begin{eqnarray}
  \label{eq:tfo}
T^\go=\frac{\mathcal A(\gamma(0))}{\int_0^{2\pi}a^\go(\theta) d\theta}  
\end{eqnarray}
and the maximal curvature is bounded for times smaller than
$T^\go$.  Moreover, up to time $T^\go$ the curve is $C^\infty$ \cite{cf:Angenent1}.


Set for every $K\ge 0$
\begin{equation}\begin{split}
\label{eq:T}
T_K:=\sup\{t \ | \  \limsup_{\go\to0} \sup_{t'\le t} \| k^\go(\cdot  , t')\|_\infty<K/2 \},\quad
T^*:=\lim_{K\to \infty} T_K.
\end{split}\end{equation}
Note that $T_K>0$ for all $K\ge K_0(\gamma(0))$. Indeed, 
if $g^\go=a^\go k^\go$ (cf. \eqref{eq:g}), one has from \eqref{eq:dg},
using $\partial_s\theta=k$ and dropping the argument $\go$
\begin{eqnarray}
  \label{eq:dg2}
   \partial_t g&=&a^2k^3+\partial_s(a\partial_s g)
=\frac1a g^3+a\partial^2_s g+k\,\partial_\theta a\, \partial_s g.
\end{eqnarray}
From \eqref{eq:dg2}, recalling that $\partial_\theta a$ is uniformly bounded and calling $\gmv(t)$ the maximal value of $|g|$ along $\gamma(t)$, we obtain
\begin{eqnarray}
  \label{eq:dgmax}
  \frac{\dd}{\dd t}\gmv(t)\le \frac1{a_{\min}}(\gmv(t))^3
\end{eqnarray}
and since $a_{\min}$ is bounded away from zero uniformly in $\go$, we find that
$|g|$ (and therefore the curvature) cannot explode instantaneously.

\begin{theorem}
  \label{th:esistenza}
Let the initial condition $\gamma(0)$ be a Jordan curve of finite
length, whose curvature is everywhere defined and $C^\infty$ 
as a function of the arc-length coordinate. Suppose moreover that
$\gamma(0)$ has a finite number of inflection points.

Fix $K>0$. There exists a unique solution $(\gamma(t))_{t\le T_K}$ of
\eqref{eq:mc} that is a Jordan curve for $t\le T_K$. The area enclosed
by $\gamma(t)$ is $2(T(\cD)-t)$. The  curvature function is equicontinuous in the
following sense: for every $\epsilon>0$ there exists
$\delta=\delta(\epsilon,K,\gamma(0))>0$ such that if $t,t'\le T_K$ and $p\in\gamma(t),p'\in\gamma(t')$ with $|t-t'|\le \delta, |p-p'|\le \delta$
then $|k(p,t)-k(p',t')|\le \epsilon$. 
\end{theorem}
Further regularity properties of the limit flow and precise the connection
with the one-dimensional heat equation are given in Lemma \ref{analytic} below.


Note that, while $T^\go$ in \eqref{eq:tfo} converges to $T$ (cf. \eqref{eq:12})
for $\go\to0$,
it is not guaranteed that
$T^*=T$: in principle, the curvature of the regularized solution could
take larger and larger values (when $\go\to0$) at some time $T^*<T$.
The following result rules out this pathological behavior:
\begin{theorem}
\label{th:gray}
  One has $T^*=T(\cD)$ and the curve $\gamma(t)$ shrinks to a point when
  $t\to T(\cD)$.
\end{theorem}
Clearly, Theorems \ref{th:esistenza} and \ref{th:gray} imply Theorem \ref{th:existence}.

\subsection{Regularity and maximum principles}
In this section we give some additional properties of the
anisotropic  curve shortening flow \eqref{eq:mc}, in addition to those
stated in Theorem \ref{th:esistenza}. These will be crucial for the
proof of Proposition \ref{th:35}.

  Let $\gamma(t)$ satisfy equation \eqref{e:mceq} and consider a
  portion of $\gamma(t)$ that is the graph $y(x,t)$ of a function in the
  coordinate frame $(\be_{\theta_0}, \be_{\theta_0+\pi/2})$ obtained
  by rotating by $\theta_0$ anti-clockwise the usual Cartesian
  frame. The curvature is given by
\begin{eqnarray}
  \label{eq:10bis}
  k(x,t)=\frac{\partial^2_x y(x,t)}{(1+(\partial_x y(x,t))^2)^{3/2}}.
\end{eqnarray}
The evolution of $y$, of the angle $\Theta=\arctan(\partial_xy)$ and of the
``curvature'' $g=Ak$ are given by the parabolic equations
 \begin{equation}\label{e:y}
  \partial_t y:= \partial^2_xy\cos^2(\Theta)A(\Theta+\theta_0)
 \end{equation}
\begin{eqnarray}
  \label{eq:angolo}
  \partial_t \Theta=\cos^2(\Theta)\partial_x(A(\Theta+\theta_0)\partial_x \Theta)
\end{eqnarray}
\begin{equation}\label{e:g}
 \partial_t g=\cos^2(\Theta)\partial_x(A(\Theta+\theta_0)\partial_x g)+\frac{g^3}{A(\Theta+\theta_0)},
\end{equation}
 \begin{equation}\label{e:k}
\partial_t k=\cos^2(\Theta)\partial^2_x( A(\Theta+\theta_0)k)+A(\Theta+\theta_0)k^3.
\end{equation}

Equation \eqref{e:y} is obtained just by projecting Equation
\eqref{e:mceq} in the chosen frame of coordinates. For
\eqref{eq:angolo} and \eqref{e:g}, one just needs to compute the
derivatives with some patience (\eqref{eq:angolo} is given also in
\cite[Equation (1.2)]{ChuZu} with a different notation 
).
Equation \eqref{e:g} can be derived from \eqref{eq:dg}
by noticing that
$$\partial_t|_{\text{along flow lines}}=\partial_t|_{\text{for fixed
    $x$}}-\frac{\partial_x y Ak}{\sqrt{1+(\partial_x y)^2}}\partial_x \text{ and }
\partial_s=\frac{1}{\sqrt{1+(\partial_x y)^2}}\partial_x.$$
Hence 
\begin{equation}
 \partial_t g(x,t)=  \frac{1}{\sqrt{1+(\partial_x y)^2}}\partial_x \left(\frac{A}{\sqrt{1+(\partial_x y)^2}}\partial_x g\right) +A^2 k^3 + \frac{A \partial_x y \partial^2_x y }{(1+(\partial_x y)^2)^2}\partial_x g,
\end{equation}
and a line of computation leads to \eqref{e:g}. 




\begin{lemma}[Properties of the regularized evolution]\label{unpointneuf}
Let $\gamma(t)$ be a solution of \eqref{e:mceq} with  $A(\cdot)$ smooth.
\begin{itemize} 
\item[(i)] For a given choice of Cartesian coordinates $(\be_{\theta_0}, \be_{\theta_0+\pi/2})$, local minima
  for $y$ and $\theta$ strictly increase with time and local maxima
  strictly decrease with time.
 \item[(ii)] The total curvature of an arc connecting two isolated inflection points is decreasing with time, and the $\theta$-intervals of tangent directions
 of such arc (angle span) strictly nest with time.
\item[(iii)] Local minima of $|g|$ or $|k|$ that are not inflection points are
  strictly increasing with time.  In particular, isolated points with
  curvature zero that are not inflection points disappear
  instantaneously. Such points can arise only when two or more inflection 
points merge.
 \item[(iv)] If $t>0$, $\gamma(t)$ does not contain flat pieces and 
the number of inflection points of the curve is non-increasing with time.
 \end{itemize}
\label{lemma:lourd1}
\end{lemma}

\begin{proof}
  Point $(i)$ is a consequence of the maximum principle
  \cite[Ch. 7]{Evans} applied to equation \eqref{e:y} and
  \eqref{eq:angolo}. These equations are non-linear, but since we know
  a priori that $\gamma(t)$ is $C^\infty$, one can treat e.g.  the
  factor $\cos^2(\Theta)A(\Theta+\theta_0)$ in \eqref{e:y} as a
  smooth positive coefficient.
Point $(ii)$ is a consequence of $(i)$, since inflection points are
local maxima or minima of $\theta$ (see also Lemma 1.6 in
\cite{ChuZu}). Point $(iii)$ is obtained from \eqref{e:g} and \eqref{e:k}.  Say for
instance that one has a local minimum of $g$ that is non-negative. Then,
omit the term $(1/A) g^3$ that is locally positive and apply the maximum principle to the remaining equation.  Point
$(iv)$ is proven in Lemma 1.4 of \cite{ChuZu} and is based on a
theorem by Angenent \cite{Angenent_zeros} that says that the set of zeros of a
parabolic equation is finite at all positive times and its cardinality is
non-increasing with time.
\end{proof}

When the regularization parameter $\go$ tends to zero, Theorem 
\ref{th:esistenza} says that \eqref{e:y} holds, with $A(\cdot)$ replaced
by $a(\cdot)$,
but the analog of Equations \eqref{eq:angolo} to \eqref{e:k} is
not guaranteed to hold. However, the following result gives extra information 
with respect to Theorem \ref{th:esistenza}:
 \begin{lemma}[Properties of the limit evolution]\label{analytic}
  Let $\gamma(t)$ solve \eqref{eq:mc}. 
  \begin{itemize}
  \item [(i)] In the Cartesian   frame $(\be_{\theta_0}, \be_{\theta_0+\pi/2})$
  with $\theta_0\in\{ \pi/4, 3\pi/4,5\pi/4,7\pi/4\}$, 
  along the portions of $\gamma(t)$  that are locally the graph of a $1$-Lipshitz function $y(x,t)$,
  the evolution of $y$ is given by the one-dimensional heat equation 
  \begin{equation}
   \partial_t y=\frac14 \partial^2_x y
\label{e:heat}
  \end{equation}
(which explains the title of the article).
\item [(ii)] For any positive time, $\gamma(t)$ consists of a
  finite number of analytic arcs separated by points where the tangent
  angle belongs to $\{0, \pi/2, \pi, 3\pi/2\}$.
  \item[(iii)] The total curvature of an arc connecting two isolated inflection points is decreasing with time, and the $\theta$-intervals of tangent directions
 of such arc strictly nest with time. The 
 number of inflection points is decreasing with time (they can merge).
\item [(iv)]  Flat portions of $\gamma(t)$ with  zero curvature
  disappear instantaneously and can in principle be present only at a
  finite number of times $t_i$. 
\item [(v)]
 Inflection points evolve continuously in time (except possibly at 
the times $t_i$ of point $(iv)$).
\item [(vi)] Isolated points where $k=0$, that are not inflection
  points, can be present only at a finite set of times and disappear instantaneously.
  \end{itemize}
 \end{lemma}
 \begin{proof}
We emphasize that this proof assumes the existence and regularity
results of Theorem \ref{th:existence}.

For point $(i)$, just look at \eqref{e:y} and \eqref{eq:a} and observe that, for $\theta_0\in\{\pi/4,3/4\pi,5/4\pi,7/4\pi\}$ and $|u|\le1$, one has
\[
\frac1{1+u^2}a(\arctan(u)+\theta_0)=\frac14.
\]

Fix $t_0$, consider any point $p\in\gamma(t_0)$ where the tangent
angle is not a multiple of $\pi/2$ and rotate the Cartesian coordinate frame
by a suitable multiple of $\pi/4$ so that locally the curve is the
graph of a $1$-Lipschitz function. Call $p_1$ the horizontal projection of $p$ in
this frame. There exists an open interval $I\ni p_1$ and $t_1>t_0$
such that the evolution of the curve is described by the solution
$y(x,t)$ of the heat equation \eqref{e:heat}, for $x\in I$ and $t\in
[t_0,t_1)$, with time-dependent boundary conditions at the endpoints
of $I$.  Then \cite[Theorem 10.5.1]{cannon} tells us that this solution is
analytic in $x$ in the interior of the domain $I$, for $t_0<t<t_1$.
This guarantees that the curve is composed of analytic arcs whose
evolution in a certain coordinate frame is given by the heat
equation  \eqref{e:heat}. They are separated either by points or by flat segments
where the tangent angle is a multiple of $\pi/2$. Point $(ii)$ is proven.

\smallskip

For point $(iii)$ it is sufficient to prove that local minima of
$\theta$ are increasing, and local maxima are decreasing (this also implies that the number of inflection points cannot increase).
We remark that in the neighborhood of an extremum of $\theta$ (i.e. of
an inflection point) the curve can always be represented as the graph of a 
$1$-Lipshitz function $y(x)$  in the Cartesian   frame $(\be_{\theta_0}, \be_{\theta_0+\pi/2})$
  with $\theta_0\in\{ \pi/4, 3\pi/4,5\pi/4,7\pi/4\}$
(this is the case even when the tangent angle is a multiple of $\pi/2$).
The inflection point corresponds to a local extremum of $\partial_xy$. As $\partial_xy$ is the solution of the heat equation, the result follows.

\smallskip

If $\gamma(0)$ contains flat segments of slope not multiple of $\pi/2$, then
by points $(i)$-$(ii)$ they disappear instantaneously and cannot be  created at later times.
Flat segments with slope multiple of $\pi/2$ require a different argument.
If such a segment  (say, with slope $0$) is present at some time $t_0$, then
\begin{itemize}
\item either it is  locally a maximum/minimum for the curve in the
  usual Cartesian frame (see Fig. \ref{fig:partipiatte} (a)) and in that case it disappears immediately since
the equation \eqref{e:y} is strictly parabolic;
\item or it is not (see Fig. \ref{fig:partipiatte} (b)), in which case its evolution in the
  Cartesian frame rotated by $\pm\pi/4$ is locally described by
  \eqref{e:heat}, so  for $t$ just after $t_0$ the curve is analytic and cannot contain
flat segments.
\end{itemize}

The only way to create a flat segment (say with slope zero) is to have an analytic arc, connecting two isolated points of zero slope (``poles'') that degenerates at some time $t_i$ to a 
horizontal segment, see Fig. \ref{fig:partipiatte} (c). At time $t_i$, the number of poles 
in $\gamma(t)$ therefore decreases by $2$. However, poles cannot be created
(again because \eqref{e:y} is strictly parabolic), so this scenario can 
 happen only a finite number of times ($\gamma(0)$ has a finite number
 of poles because we assumed it has a finite number of inflection points).
This proves point $(iv)$.

\begin{figure}[hlt]
 \begin{center}
 \leavevmode 
 \epsfxsize =10  cm
 \psfragscanon
 \psfrag{a}{$(a)$}
\psfrag{b}{$(b)$}
\psfrag{c}{$(c)$}
\psfrag{tinf}{$t<t_i$}
 \psfrag{teg}{$t=t_i$}
 \epsfbox{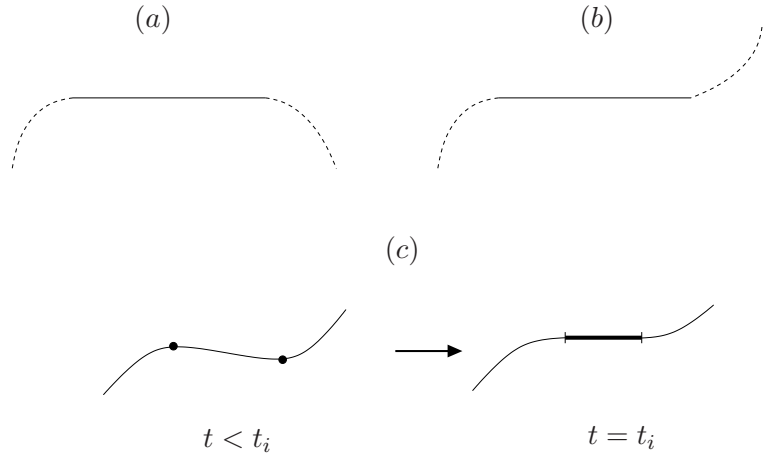}
 \end{center}
 \caption{Examples of flat segments on the curve corresponding to a local extremum in the local coordinate frame (a) and not corresponding to a local extremum (b).
 The third picture (c) describes the only possible way to create flat segments on $\gamma(t)$ at positive time }
 \label{fig:partipiatte}
 \end{figure}


For point $(v)$, continuity of the curvature in space and time
obtained in Theorem \ref{th:esistenza} makes it impossible for the
inflection points to jump, except possibly at times $t_i$ (at such
times, the position of the inflection point on the flat segment is not
well defined anyway). 

For point $(vi)$, let us call for simplicity ``zero-curvature points''
the points in question. It is immediate (from  the
heat equation) that zero-curvature points with slope not multiple of
$\pi/2$ that are present in $\gamma(0)$ disappear immediately and
cannot be created at later times. A zero-curvature point with (say) 
slope zero is a pole. If such a point existed for a time interval
$[t_a,t_b)$, the pole would not move in $[t_a,t_b)$ (its velocity
would be
zero according to the curve-shortening flow), which is not
possible since the equation \eqref{e:y} is strictly parabolic. At
later times, zero-curvature points of zero slope can appear only when
an odd number of poles merge (see Fig. \ref{fig:k0}): this can happen only a
finite number of times.  Isolated poles cannot become
zero-curvature points, as can be easily obtained from the fact that,
for the regularized evolution, local minima of $k$ where $k>0$ are
non-decreasing, cf.  point $(iii)$ of Lemma \ref{lemma:lourd1}.

\begin{figure}[hlt]
 \begin{center}
 \leavevmode 
 \epsfxsize =10  cm
 \psfragscanon
 \psfrag{eps}{$\epsilon$}
\psfrag{eps2}{$\epsilon$}
\psfrag{D}{$\mathcal{D}$}
\psfrag{F}{$\mathcal{F}$}
 \psfrag{G}{$\mathcal{G}$}
 \psfrag{H}{$\mathcal{H}$}
 \epsfbox{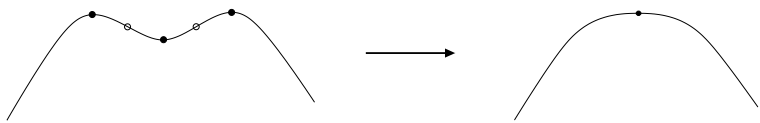}
 \end{center}
 \caption{ Three poles (black dots) merge to
   create a single pole (right drawing). At the time of the merging,
   the curvature at the new pole is zero: this is because the
   curvature is zero at the two inflection points (white dots) and
   curvature is continuous w.r.t. space and time (Proposition
   \ref{th:continua}).}
 \label{fig:k0}
 \end{figure}


\end{proof}

\section{Proof of Theorem \ref{th:esistenza}}
The scheme of the proof is the following. First we prove that
$\gamma^\go(t)$ has a limit for every fixed $t$, when $\go$ tends to
zero (this is done in Section \ref{sec:limflow}). Then in Sections
\ref{sec:conclu1} and \ref{sec:conclu} 
we prove that $\gamma(t)$ satisfies the regularity properties stated
in Theorem \ref{th:esistenza}, and that it does solve equation \eqref{eq:mc}.

\subsection{Some preliminary remarks}

Given a smooth simple curve $\gamma$ we write $\gamma+x{\bf N}$ for the curve
obtained by shifting each point of $\gamma$ by an amount $x\in \mathbb
R$ in the (inward) normal direction at $x$.
Set
\begin{eqnarray}
  \label{eq:13}
 m(\gamma):=\sup\{\eta\,: \forall |x|\le \eta, \gamma+x{\bf N}
\text{ is a simple curve }\}. 
\end{eqnarray}
\begin{remark}
  If $|x|<m(\gamma)$, then
the function $\gamma\mapsto \hat \gamma=\gamma+x{\bf N}$ defined by
 $ \gamma\ni p \mapsto
p+x{\bf N}\in\hat\gamma$ is a bijection that 
preserves the normal direction $\bf N$
As a consequence, if $\mathcal D$ is the interior of $\gamma$, then
the interior of $\hat \gamma$ is $\mathcal D^{(-x)}$.
Identifying the points that are in bijection, we have that the
curvatures of $\gamma$ and $\hat \gamma$ are related by
\begin{eqnarray}
  \label{eq:6}
\hat k=\frac{k}{1-kx}.
\end{eqnarray}
\label{rem:bij}
\end{remark}

\begin{lemma}\label{ouille}
For the evolution $\gamma^\go(t)$ set
\begin{equation}
\label{eq:43}
 u^\go(t):=\min\{ |x-x'|>0 \ : \ x\ne x'\in \gamma^\go(t)\text{ such that }   |x-x'| \text{ is a local minimum}\}
\end{equation}
and $r^\go(t)$ to be the minimal value of the curvature radius $|k^\go(p)|^{-1}$ for $p\in \gamma^\go(t)$.
Then 
\begin{equation}
m(\gamma^{\go}(t))=\min (u^\go(t)/2, r^\go(t)).
 \end{equation}
 Furthermore $t\mapsto u^\go(t)$ is an increasing function.
\end{lemma}
\begin{remark}
What is meant precisely in the definition of $u^\go(t)$ is that $x,x'$
are such that for all sufficiently small $\delta>0$, the minimum 
$|y-y'|$ with $y\in \gamma^\go(t)\cap \cB(x,\delta), y'\in
\gamma^\go(t)\cap \cB(x',\delta)$
is positive and realized for $y=x,y'=x'$.
If the set in \eqref{eq:43} is empty as in the case where $\gamma$ is convex then then we set $m(\gamma)$ to be infinite.
\end{remark}

\begin{proof}[Proof of Lemma \ref{ouille}] We drop all superscripts
  $\go$ in this proof.
Let us consider the evolution of $\gamma(t)+x {\bf N}$ for fixed $t$ when $x\ge 0$ increases.
There are two ways for the curve to create self-intersections: either
the maximal curvature explodes  when $x\to x_0^-$ and a loop is created when $x> x_0$, 
or two distant arcs of the curves kiss when $x=x_0$.
The first can occur when $x=r(t)$ if $r(t)$ is attained at a point of positive curvature and the second for $x=u(t)/2$ if the chord linking the two points
where $u(t)$ is attained is enclosed in $\gamma$. A
similar argument holds for the evolution
of $\gamma(t)+x {\bf N}$ for fixed $t$ when $x\le 0$ decreases; hence our result.

\smallskip

To see that $u(t)$ is increasing, fix $t$ and consider two points $p_1,p_2\in\gamma(t)$ realizing the minimum
$u(t)$. 
Assume to fix ideas that the segment connecting them
is vertical, enclosed in $\gamma$ and that $p_1$ is above $p_2$.  If $k_{p_1}, k_{p_2}$ are the curvatures at $p_1,p_2$, then the
definition of $u(t)$ (the distance being locally minimized by
$p_1,p_2$) implies that
\begin{eqnarray}
  \label{eq:19}
k_{p_i} u(t)\le 1, i=1,2.   
\end{eqnarray}


Let
$\Gamma(t,s):=\gamma(t+s)+{\bf N}u(t)/2$ (cf. Remark \ref{rem:bij}, see Fig/ \ref{fig:uff} for a graphical construction) and consider the points
$p_1(t,s),p_2(t,s)$ on $\Gamma(t,s)$ 
with same horizontal coordinate as $p_1,p_2$, that tend to $p_1,p_2$ when
$s\to0$. For short times $s$, locally around $p_1(t,s)$ and $p_2(t,s)$, the
curve $\Gamma(t,s)$  
is the graph of
functions $y_1(x,t,s)$ and $y_2(x,t,s)$. For $s=0$ one
has $$y_1(\cdot,t,0)\ge y_2(\cdot,t,0)$$
by definition of $u(t)$.
To prove that $u(t)$ is increasing, we just need to show that this
inequality remains valid (locally around $p_i(t,s)$) for small positive $s$.

\begin{figure}[hlt]
 \begin{center}
 \leavevmode 
 \epsfxsize =13.5  cm
 \psfragscanon
 \psfrag{u}{$u(t)$}
\psfrag{u2}{$\frac{u(t)}{2}$}
\psfrag{p1}{$p_1$}
\psfrag{p2}{$p_2$}
\psfrag{gt}{$\gamma(t)$}
\psfrag{gts}{$\gamma(t+s)$}
\psfrag{vts}{$\Gamma(t,s)$}
\psfrag{pts1}{$p_1(t,s)$}
\psfrag{pts2}{$p_2(t,s)$}
\psfrag{G0}{$\Gamma(t,0)$}
\epsfbox{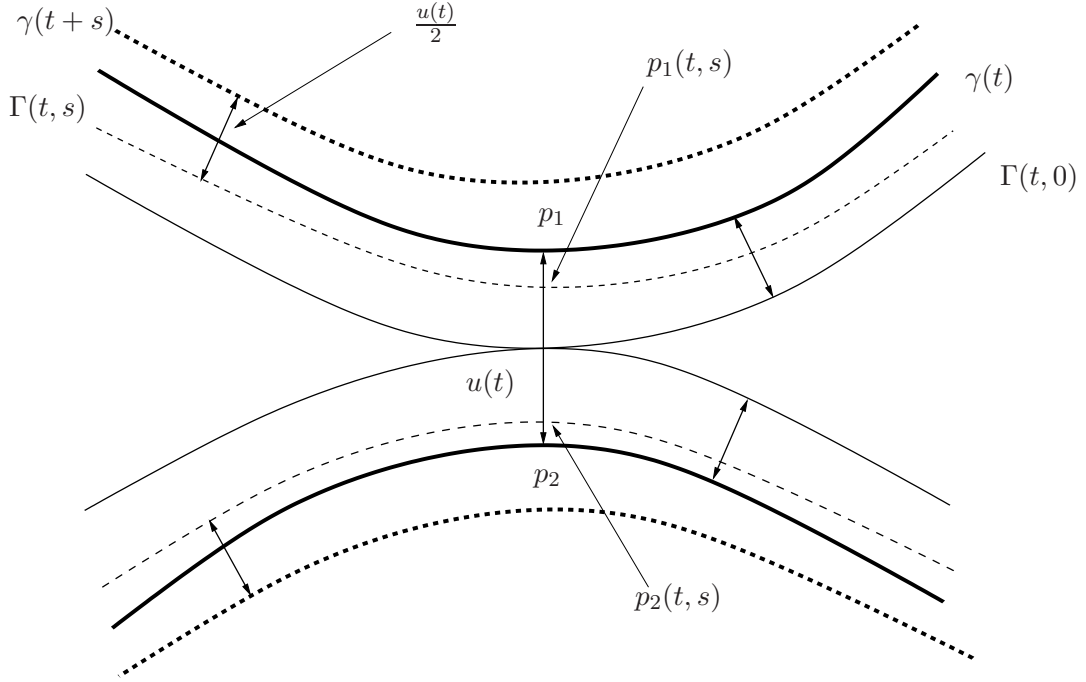}
 \end{center}
 \caption{Graphical representation of the curve $\Gamma(t,s)$ around the points $p_1$ and $p_2$}
 \label{fig:uff}
 \end{figure}

From Remark \ref{rem:bij}
 one has for fixed $t$
$$\partial_s \Gamma
=a(\theta)
\frac{\hat k}{1+\hat k\,u(t)/2} {\bf N}=
a(\theta)
\left(\hat k-\frac{u(t)\hat k^2}{2+u(t)\hat k}\right) {\bf N}$$
with $\hat k$ the curvature of $\Gamma$.
The reason is that, if a point of $\Gamma(t,s)$ has slope $\theta$ and curvature
$\hat k$, the point in $\gamma(t+s)$ that is in bijection with it has the
same slope $\theta$ and curvature $\hat k/(1+\hat k u(t)/2)$, see \eqref{eq:6}
with $x=u(t)/2$.
 See also Lemma \ref{lemma:inclu} and in
particular \eqref{splitto} (where one has to take $\gep'(t)=0$) for
the more detailed proof of a similar result.

Projecting this equation vertically (like what we did to get
\eqref{e:y}) we obtain  differential equations for $y_1$ and $y_2$ and
(keep in
mind that the normal vector ${\bf N}$ is pointing downward around $p_1(t,s)$
and upward around $p_2(t,s)$) we deduce that
$z=y_1-y_2$ satisfies the parabolic equation
\[
\partial_s z(x,t)=b_1 \partial^2_x z(x,t)+b_2 \partial_x z(x,t)+b_3
\]
where
$b_1=f(\partial_xy_1)=(1+(\partial_xy_1)^2)^{-1}a(\arctan(\partial_xy_1))>0$
and $b_2=
\partial^2_xy_2(f(\partial_xy_1)-f(\partial_xy_2))/(\partial_xy_1-\partial_xy_2)$ are two smooth functions
and 
$$b_3=u(t)\left[
 (1+(\partial_xy_1)^2)^{1/2} \frac{a(\arctan(\partial_xy_1))\hat k_1^2}{2+\hat
    k_1u(t)}+(1+(\partial_xy_2)^2)^{1/2}\frac{a(\arctan(\partial_xy_2))\hat k_2^2}{2+\hat k_2u(t)}
\right].$$
Here, $\hat k_i, i=1,2$ denotes the curvature along the portions of
the curve $\Gamma$ 
whose graph is $y_i$. If we can prove that $b_3\ge0$, then from the
(weak) maximum principle it follows that locally $z\ge 0$ for positive time and
thus the result is proved.
To see that $b_3$ is positive, note that (again thanks to
\eqref{eq:6}),
\begin{eqnarray}
  \label{eq:20}
\frac1{2+\hat k_1 u(t)}=\frac12\left(1-\frac{u(t) k_1}2\right)  
\end{eqnarray}
with $ k_1$ the curvature on $\gamma(t+s)$ of the point in
bijection with the point of $\Gamma(t,s)$ with curvature $\hat k_1$.
For $s=0$ and looking at point $p_1$, the r.h.s. of \eqref{eq:20} is
larger than $1/4$ thanks to \eqref{eq:19}, since $k_1$ reduces to
$k_{p_1}$. Thanks to continuity of curvature in space and time
(the regularized evolution is $C^\infty$), the positivity continues to
hold for $s$ small, close to $p_1(t,s)$. A similar argument shows that
the second term in $b_3$ is positive.
\end{proof}
\subsection{The limit flow}
\label{sec:limflow}
We first show that, for $t<T^*$ (cf. definition \eqref{eq:T}), $\gamma^\go(t)$ converges to a limit curve
as $\omega\to0$:
\begin{proposition}
\label{prop:cauchy} Take any sequence $\omega_n\to0$.
For any fixed $K$,  $\gamma^{\go_n}(t)$ 
is a Cauchy sequence for the Hausdorff distance $d_{\mathcal H}$, uniformly for $t\le T_K$.
\end{proposition}
\begin{proof}[Proof of Proposition \ref{prop:cauchy}]
Let us consider $\go$ and $\go'$ close enough to zero so that
the curvature of $\gamma^\go(t),\gamma^{\go'}(t)$ is bounded by $K$ uniformly on $[0,T_K]$. 

\begin{lemma}\label{lemma:inclu}
Let $\cD^\go(t)$, $\cD^{\go'}(t)$ be the compact sets enclosed by the curves $\gamma^\omega(t),\gamma^{\omega'}(t)$.
Consider $\eta<(1/2)\min(1/K, m(\gamma(0)))$ and set
$$\gep(t):=\eta \exp(10 K^2(t-T_K)).$$
For all $t\le T_K$, for all $\go,\go'$ sufficiently small,
the curves $\gamma^\go(t) \pm \gep(t){\bf N}$ are simple curves
and are the boundaries of $\left(\cD^\go(t)\right)^{(\mp \gep(t))}$
(cf. Remark \ref{rem:bij}). 
Furthermore
\begin{equation}\label{inclu}
\left(\cD^\go(t)\right)^{(-\gep(t))} \subset \cD^{\go'}(t)\subset \left( \cD^{\go}(t)\right)^{(\gep(t))}.
\end{equation}
\end{lemma}
\begin{remark}
  We will see that it is sufficient to choose $\omega,\omega'$ sufficiently small so that they satisfy
  \begin{eqnarray}
    \label{eq:condiz}
    \|a^\omega-a^{\omega'}\|_\infty\le K\eta \exp(-10 K^2 T_K)
  \end{eqnarray}
with $a^\go(\cdot)$ the regularized version of $a(\cdot)$.
\end{remark}
As a consequence of \eqref{inclu} and of $\gep(t)\le \eta$, 
$$ \max_{t\le T_K} d_{\mathcal H}(\gamma^{\go}(t),\gamma^{\go'}(t))\le \eta$$
and Proposition \ref{prop:cauchy} follows.
\end{proof}

\begin{proof}[Proof of Lemma \ref{lemma:inclu}]

  To prove that $\gamma^\go(t) \pm \gep(t){\bf N}$ are simple curves
  for all $t\le T_K$, it is sufficient to remark that by Lemma
  \ref{ouille}
$$m(\gamma^\go(t))\ge \min (u(0)/2,1/K)\ge \eta \ge \gep(t).$$
We prove the inclusions \eqref{inclu}.  Set 
$$
\gamma^-(t):=\gamma^\go(t)+ \gep(t){\bf N} \text{ and } 
\gamma^+(t):=\gamma^\go(t)-\gep(t){\bf N}$$
and observe that $\gamma^\pm(t)$ is the boundary of $(\mathcal
D^\go(t))^{(\pm \gep(t))}$, see Remark \ref{rem:bij}.
Let $ k^\pm$ be the
respective curvature functions of $ \gamma^\pm$, that are
related to $k^\go$ via $1/k^\pm =1/k^\go\pm \gep(t)$, cf. \eqref{eq:6}.

The normal velocity of a point on $\gamma^-$ is given by
$$  v^-= a^\omega(\theta)k^\omega{\bf N} +\gep'(t){\bf N},$$
with $\gep'(t)$ the $t$-derivative of $\gep(t)$.
Similar considerations concerning $ \gamma^+$ implies that the curves $ \gamma^\pm$ are solution of the following 
equations:
\begin{equation}\begin{split}
\partial_t\gamma^-= v^-& = a^\go (\theta)  k^-{\bf N}+ \left(\gep'(t)-\frac{a^\go(\theta)\gep(t)( k^-)^2}{1+\gep(t) k^-}\right){\bf N},\\ 
  \partial_t\gamma^+= v^+ &= a^\go (\theta)  k^+{\bf N}- \left(\gep'(t)
-\frac{a^\go(\theta)\gep(t)( k^+)^2}{1-\gep(t) k^+}\right){\bf N}. 
\label{splitto}
\end{split}\end{equation}

The curve $ \gamma^-(t)$ is clearly enclosed by $\gamma^{\go'}(t)$ at
initial time, since 
$\gamma^{\go'}(0)=\gamma(0)$.
To show that this inclusion holds for all times,
observe that the normal velocity associated to the evolution of
$\gamma^-(t)$ can be written  
\[
v^-=a^{\go'}(\theta)k^-{\bf N}+{\bf N}\Delta:=a^{\go'}(\theta)k^-{\bf
  N}+{\bf N}\left((a^{\go}(\theta)-a^{\go'}(\theta))k^-+
\gep'(t)-\frac{a^\go(\theta)\gep(t)( k^-)^2}{1+\gep(t) k^-}
\right).
\]
If $\Delta$ were zero, $\gamma^-(t)$ would just evolve with anisotropy
function $a^{\go'}$ (like $\gamma^{\go'}(t)$), so it would be a
standard fact that $\gamma^-(t)$ would be included in
$\gamma^{\go'}(t)$ for all times. This clearly remains true if we can
show that $\Delta\ge0$, since the term ${\bf N}\Delta$ gives a further inward push
to $\gamma^{\go'}(t)$.


As $|k^\go|\le K$ for all $t\le T_K$ and $\gep(t)\le 1/(2K)$, from \eqref{eq:6} we have
$-K\le  k^- \le 2K$. Hence 
\begin{multline}
 \Delta
 \ge 
 -2K|a^{\go}-a^{\go'}|(\theta)+\gep(t)
\left[10K^2-\frac{a^\go(\theta)4K^2}{1-\eta K}\right]\\
 \ge  2K^2 \eta \exp(-10 K^2T_K)-2K\|a^{\go}-a^{\go'}\|_{\infty}\ge 0,
\end{multline}
where in the second inequality we used that $\eta K\le \frac{1}{2}$ and that $a^\go\le 1$, and in the last one  we used condition \eqref{eq:condiz}.


With a similar reasoning, the proof that $\gamma^{\go'}(t)$ is enclosed
by $\gamma^+(t)$ can be reduced to check that
\begin{equation}
a^\go (\theta)  k^+-\left(\gep'(t)-\frac{a^\go(\theta)\gep(t)( k^+)^2}{1-\gep(t) k^+}\right)\le a^{\go'}(\theta)  k^+
\end{equation}
whose proof is analogous to that of $\Delta\ge0$.

\end{proof}

\subsection{Uniform regularity of the (regularized) flow}
\label{sec:conclu1}
Here we prove that the curvature function is an equicontinous function of
arc-length, uniformly in $\omega$, up to time $T_K$.
This is an essential step in establishing the regularity properties
stated in Theorem \ref{th:esistenza}, point $(ii)$.

Recall the definition of $g^\go$ given just before \eqref{eq:dg2}.
\begin{proposition}
\label{prop:holder}
  Given $K<\infty$ 
there exists a
  constant
$f(K,\gamma(0))$ 
such that, for every $t\le T_K$ and $\go\in (0,1)$,
\begin{eqnarray}
  \label{eq:regolarita}
|g^\go(s_1,t)-g^\go(s_2,t)|\le
\sqrt{|s_2-s_1|} f(K,\gamma(0)).
\end{eqnarray}
\end{proposition}
This $1/2$-H\"older regularity is not optimal (one
can improve it to $1^-$-H\"older regularity) but is sufficient
for our purposes.

\begin{proof}[Proof of Proposition \ref{prop:holder}]
It is sufficient to show that (omitting for simplicity the argument
$\omega$ everywhere)
\begin{eqnarray}
  \label{eq:1}
\psi(t):=\int_{\gamma(t)} \left(\partial_s g\right)^2ds\le f(K,\gamma(0))^2
  \end{eqnarray}
 (we will see that  the dependence of $f$ on $\gamma(0)$ is only
  through the value $\psi(0)$).  Indeed, \eqref{eq:1} and
  Cauchy-Schwarz imply
\begin{eqnarray}
|g(s_1,t)-g(s_2,t)|=\left|\int_{s_1}^{s_2} \partial_s g \;ds\right| \le 
\sqrt{|s_2-s_1|} f(K,\gamma(0)).
\end{eqnarray}

To show \eqref{eq:1}, write
\begin{eqnarray}
  \label{eq:2}
\frac{\dd}{\dd t}\psi(t)&=&\frac{\dd}{\dd t}  \int_{\gamma(t)} \left(\partial_s g\right)^2ds=
 \int_{\gamma(t)} \partial_t\left(\partial_s g\right)^2ds-
 \int_{\gamma(t)} ak^2 \left(\partial_s g\right)^2ds,
\end{eqnarray}
where the second term comes from the fact that the length of $\gamma$ 
decreases with time according to \eqref{eq:dleng} (see also the proof of Proposition 2.7 in
\cite{Grayson} for the formula in the case $a\equiv 1$)
and $a$ should be seen as $a(\theta(s))$, with $\theta(s)$ the angle at the point
of $\gamma(t)$ with arc-length coordinate $s$.
Applying \eqref{eq:dg} and \eqref{eq:7} 
with $A(\cdot)$ replaced by $a(\cdot)=a^\go(\cdot)$, the right-hand side of \eqref{eq:2} can be rewritten
as
\begin{eqnarray}
\label{eq:3termini}
\int_{\gamma(t)} a k^2 (\partial_s g)^2 ds+2 \int_{\gamma(t)} \partial_s g \partial_s(a^2
k^3)+
2\int_{\gamma(t)} (\partial_s g)\partial^2_s(a\partial_s g)ds.
\end{eqnarray}
  The first term is bounded
by $C(K) \psi(t)$. 
For the second one,
we observe that, as $  \partial_s\theta=k$,
\[
\partial_s(a^2k^3)=3 a k^2\partial_s g-ak^3\partial_s a=3 a
k^2\partial_s g-ak^4 \partial_\theta a.
\]
Altogether, recalling that 
$a^\go$ is Lipschitz uniformly in $\go$, the second term is bounded by 
$C(K)[\psi(t)+\sqrt{\psi(t)}]$.

The third term in \eqref{eq:3termini}
looks more problematic.  However, by integration by parts it equals
\begin{eqnarray}
  \label{eq:3}
  -2\int_{\gamma(t)} (\partial^2_s g)\partial_s(a\partial_s g)ds=-2\int_{\gamma(t)} a
  (\partial^2_s g)^2ds-2\int_{\gamma(t)}k\,(\partial^2_s g)(\partial_s
  g) (\partial_\theta a)\, ds.
\end{eqnarray}
Now we use that $a$ is bounded away from zero and that
$\partial_\theta a$ is bounded
away from infinity (uniformly in $\go$): \eqref{eq:3}
is then upper bounded by 
\begin{multline}
  \label{eq:4}
  -2a_{min}\int_{\gamma(t)} 
  (\partial^2_s g)^2ds+C(K)\sqrt{\int_{\gamma(t)}
  (\partial^2_s g)^2ds}\sqrt{\int_{\gamma(t)} 
  (\partial_s g)^2ds}\\=
 -2a_{min}\int_{\gamma(t)} 
  (\partial^2_s g)^2ds+C(K)\sqrt{\int_{\gamma(t)} 
  (\partial^2_s g)^2ds}\sqrt{\psi(t)}
\\
=-2a_{min}\left[\sqrt{\int 
  (\partial^2_s g)^2ds}-C'(K)\sqrt{\psi(t)}\right]^2+C''(K)\psi(t)\le C''(K)\psi(t).
\end{multline}
Altogether we have obtained
\begin{eqnarray}
  \label{eq:5}
  \frac{\dd}{\dd t}\psi(t)\le c(K)[\psi(t)+\sqrt{\psi(t)}]
\end{eqnarray}
and, since $\psi(0)<\infty$ ($\gamma(0)$ was assumed to have a
$C^\infty$ curvature) we are done.
\end{proof}




\subsection{Proof of Theorem \ref{th:esistenza}: conclusion }
\label{sec:conclu}
Uniqueness is trivial, so 
we concentrate on existence.
Proposition \ref{prop:cauchy} implies that we can define
$\gamma(t)=\lim_{n\to\infty}\gamma^{\go_n} (t)$
and actually that the limit does not depend on the chosen sequence $\omega_n\to0$.
We have to prove that $\gamma(t)$ does solve equation
\eqref{eq:mc}
and that it has the desired regularity properties.

The first step is:
\begin{proposition}\label{th:continua} Fix $K>0$.
  \begin{enumerate}[(i)]
  \item Given a sequence of points $p^\go\in \gamma^\go(t)$, for $t\le
    T_K$, that converges to a point
  $p\in\gamma(t)$, the curvature and the tangent angle of $\gamma^\go(t)$ at
  $p^\go$ converge to the curvature and tangent angle of $\gamma(t)$ at $p$.
\item The curvature function $k$  and the angle function 
$\theta$ of the family $(\gamma(t))_{t\le T_K}$ are
equicontinuous in the sense that for $\epsilon>0$ there exists 
$\delta(\epsilon,K,\gamma(0))>0$ such that if $p\in \gamma(t),p'\in\gamma(t')$
with $t,t'\le t_K$, $|t-t'|\le \delta,|p-p'|\le \delta$ then 
$|k(p,t)-k(p',t')|\le \epsilon$ and $|\theta(p,t)-\theta(p',t')|\le \epsilon$.
  \end{enumerate}
\end{proposition}
\begin{proof}[Proof of Proposition \ref{th:continua}]\

  \emph{(i) }The angle  $\theta^\go$ at $p^\go$ converges when $\go\to0$ to a limit
  $\theta$, as a simple
  consequence of the convergence in Hausdorff distance of $\gamma^\go(t)$ to
$\gamma(t)$, plus the fact that the curvature of $\gamma^\go(t)$ is
bounded by $K$. 
Assume to fix ideas (and without loss of generality) that $\theta\in
[-\pi/3,\pi/3]$ (if this is not the case, the Cartesian coordinate
frame below has to be rotated by a multiple of $\pi/2$). Then,  using Proposition \ref{prop:cauchy}:
\begin{enumerate}
\item there
exists $c=c(K)$ such that for $\go$ small enough $\gamma^\go(t)$ is
locally the graph of a
  function $x\mapsto y^\go(x,t)$ in the usual Cartesian coordinate
  frame, for $x$  in an interval
  $I=[p_1-c,p_1+c]$, where $p_1$ is the horizontal coordinate of $p$;
\item the same holds for the limit curve $\gamma(t)$ and the function  $y^{\go'}(\cdot,t)$ converges to $y(\cdot,t)$
uniformly on  $I$; more than that,
\begin{eqnarray}
  \label{eq:gagag}
  \max_{x\in I}|y^{\go}(x,t)  -y(\cdot,t)|\le \epsilon(\go,K,\gamma(0))
\end{eqnarray}
where the right-hand side does not depend on $t$ as long as $t\le
T_K$, where $\epsilon$ tends to zero when $\go\to0$;
\item 
 the $x$-derivative of $y^{\go}(\cdot,t)$ is bounded by $2$ on $I$, for
 $\go$ small enough, where the value $2$ 
 is chosen simply because $|\tan(\pi/3)|<2$.
\end{enumerate}
From \eqref{eq:10bis}, the curvature of $\gamma^\go(t)$ at the point with
horizontal coordinate $x\in I$ is
\begin{eqnarray}
  \label{eq:10}
  k^\go(x,t)=\frac{\partial^2_x y^{\go}(x,t))}{(1+[\partial_x y^{\go}(x,t)]^2)^{3/2}}
\end{eqnarray}
from which we infer that the second derivative of $y^\go(\cdot,t)$ is
uniformly bounded. Then, the Ascoli-Arzel\`a Theorem implies that 
$\partial_xy^\go(x,t)$ converges to $\partial_x y(x,t)$ uniformly on $I$, and in
particular that $\theta$ is the tangent angle of $\gamma(t)$ at $p$.
Similarly, since $k^\go$ is uniformly continuous w.r.t. arc-length
(Proposition \ref{prop:holder}), $k^\go$ converges uniformly and
the limit 
(that is a continuous function) is the curvature function of $\gamma(t)$.
Let us emphasize that 
\begin{eqnarray}
\label{eq:contvera}
 \max_{x\in I} |k^\go(x,t)-k(x,t)|\le \epsilon(\go)=\epsilon(\go,K,\gamma(0))\\
\label{eq:contveratheta}
 \max_{x\in I}   |\theta^\go(x,t)-\theta(x,t)|\le \epsilon(\go)=\epsilon(\go,K,\gamma(0))
\end{eqnarray}
i.e. the estimates are uniform in $t\le T_K$, otherwise one would easily find
a contradiction with \eqref{eq:gagag}.

\emph{(ii)} 
Take $t\le T_K$ and $p\in\gamma(t)$. As in point (i), assume without loss
of generality that $\theta\in[-\pi/3,\pi/3]$, so that the curve is
locally the graph of a function $y(x,t)$ with $x$ in some interval $I$
of width depending only on $K$.

As long as $\go>0$ we know that the evolution is smooth,  in particular 
for $t'$ very close to $t$ the curve $\gamma^\go(t')$ is still the graph
of a function in $I$ and
(see below for a bit more of detail) 
\begin{eqnarray}
  \label{eq:driftcont}
  |\partial_t \theta^\go(x,t)|\le c(\go)\;, |\partial_t k^\go(x,t)|\le c(\go)
\end{eqnarray}
for some 
$c(\go)=c(\go,K,\gamma(0))$ that may diverge as $\go\to0$.
Then, for $t'$ very close to $t$ 
\begin{align}
  \label{eq:driftcontb}
  |k^\go(x,t)-k^\go(x,t')|\le c(\go) |t-t'|\\
\label{eq:thetacontb}
|\theta^\go(x,t)-\theta^\go(x,t')|\le c(\go) |t-t'|.
\end{align}
 Assuming for a moment \eqref{eq:driftcont}
 we have from \eqref{eq:contvera}
\begin{eqnarray}
  |k(x,t)-k(x,t')|\le \inf_{\go>0} \left(2 \epsilon(\go)+c(\go)|t-t'|\right).
\end{eqnarray}
The right-hand side clearly tends to zero with $|t - t'|$ (choose a sequence
$\omega_n$
 tending to zero. If
$c(\omega_n )$ does not diverge we are done. Otherwise, compute the right-hand side for  $\go=\go_n$ with
the largest value of $n$ such that $c(\omega_n ) \le |t - t'|^{-1/2}$). 
A similar argument gives that 
\begin{eqnarray}
  |\theta(x,t)-\theta(x,t')|=o(1)\quad \text{as}\quad |t-t'|\to0
\end{eqnarray}
uniformly for $t,t'\le T_K$.
To conclude the proof of point (ii), take $p,p'$ close
to each other, with $p\in\gamma(t'),
p'\in\gamma(t)$. Just write
\[
|k(p,t)-k(p',t')|\le |k(p,t)-k(p'',t')|+|k(p'',t')-k(p',t')|
\]
with $p''$ the point on $\gamma(t')$ close to $p$ with the same horizontal coordinate as 
$p$. We have just proven that the  first term in the right-hand side is $o(1)$ as $|t-t'|$ goes to $0$; as for the second one, it vanishes when $|p-p'|\to0$, since we have shown 
in the proof of point (i) 
(using Proposition \ref{prop:holder}) that the curvature function of $\gamma(t')$ is 
uniformly continuous w.r.t. arc-length. Similarly one proves the
continuity 
statement for the angle function $\theta$.

It remains to prove \eqref{eq:driftcont}.
From 
\cite[Theorem 3.1]{cf:Angenent1}, we know that, since the regularized anisotropy function 
$a^\go(\cdot)$ is $C^\infty$, the curve $\gamma^\go(t)$ is also $C^\infty$
at all times $t<T^\go$ (cf. \eqref{eq:tfo}) and thus \eqref{eq:driftcont}
follows from \eqref{eq:dg} and \eqref{eq:dtheta}.
\end{proof}

Back to the proof of Theorem \ref{th:esistenza}, we still have to prove that the limit flow solves \eqref{eq:mc}. As in the proof of 
Proposition \ref{th:continua}, let us concentrate on a portion of
the curve that is locally described by the graph of a function
$y^\go(x,t), x\in I$. We have from \eqref{e:y} (with $\theta_0=0$), for $t_2>t_1$,
\begin{eqnarray}
  \label{eq:9}
  y^\go(x,t_2)-y^\go(x,t_1)=\int_{t_1}^{t_2}k^\go(x,u)a^\go(\theta^\go(x,u))/\cos(\theta^\go(x,u))
du,
\end{eqnarray}
with $k^\go$ as in \eqref{eq:10} and
$\theta^\go(x,u)=\arctan(\partial_x y^\go(x,u))$. From the (uniform)
convergence of $k^\go,a^\go,\theta^\go$ shown above, plus the uniform
bound
$|k^\go|\le K$ that allows to apply dominated convergence, we have for
the limit curve $\gamma(t)$
\begin{eqnarray}
  \label{eq:11}
   y(x,t_2)-y(x,t_1)=\int_{t_1}^{t_2}k(x,u)a(\theta(x,u))/\cos(\theta(x,u))
du.
\end{eqnarray}
Continuity of $k$ and $\theta$ with respect to time (Proposition
\ref{th:continua}, point \emph{(ii)}) allows to deduce 
that
\begin{eqnarray}
   \partial_t  y(x,t)=\frac{k(x,t)a(\theta(x,t))}{\cos(\theta(x,t))},
\end{eqnarray}
hence $\gamma(t)$ does solve \eqref{eq:mc}.
Finally, from \eqref{eq:darea} and the uniform convergence of
$a^\go(\cdot)$ to $a(\cdot)$, we see that the area of
$\gamma(t)$ is $\cA(\gamma(0))-t\int_0^{2\pi}a(\theta)d\theta=\cA(\gamma(0))-2t$.

\section{Proof of Theorem \ref{th:gray}}

The basic ingredients
of the proof
are the following two results, that 
say that when the curvature diverges, it necessarily does so on 
arcs of total curvature at least $\pi$.
 This says that the curve
cannot develop ``corners''.
\begin{proposition}
\label{th:21} Let $T^*$ be as in \eqref{eq:T}. For any given $K>0$ 
there exists $t< T^*$ and an arc of $\gamma(t)$ of total curvature at least $\pi$
on which the curvature has a constant sign and is
larger than $K$ in absolute value.
\end{proposition}
\begin{proposition}
  \label{th:35}
  There exists a constant $\tilde K$ (which depends on $\gamma(0)$)
  such that, for $t\le T^*$, if the curvature at $p\in \gamma(t)$ is
  larger than $\tilde K$ in absolute value then $p$ belongs to an arc
  on which the curvature has constant sign, and whose total curvature
  is at least $\pi$.
  
  \end{proposition}

Propositions \ref{th:21} and \ref{th:35} are  the analog of Theorem 2.1
and Lemma 3.5 of \cite{Grayson}; however, due to anisotropy
($a(\cdot)\not\equiv1$ in our case) and to the need to regularize
$a(\cdot)$ to $a^\go(\cdot)$ while obtaining bounds that are uniform in $\go$,
the proofs require many non-trivial modifications. This is done in
detail in Sections
\ref{sub:21} and \ref{sub:35} below. 

Given these two propositions, the proof of Theorem \ref{th:gray}
becomes essentially identical to the proof of the main theorem of
\cite{Grayson} (that says that under isotropic curve shortening flow
($a(\cdot)\equiv1$) Jordan curves do not develop singularities before
shrinking to a  point). Let us just remind the general strategy 
(giving details would amount to repeating the argument of \cite{Grayson}).

\medskip

Let $\Theta$ be the supremum of the angles $\theta$ such that
there exist sequences $t_i\to T^*$, $\epsilon_i\to0$, 
$\theta_i\to\theta$ and arcs
 $C_i\subset \gamma(t_i)$ satisfying:
\begin{itemize}
\item the diameter of $C_i$ is at most $\epsilon_i$;
\item the curvature $k$ has constant sign on $C_i$ and is pointwise
  larger than $\tilde K$ in absolute value, with $\tilde K$ the
  constant of Proposition \ref{th:35};
\item the total curvature $\int_{C_i}|k|ds$ is $\theta_i$.
\end{itemize}
 
Proposition \ref{th:21} implies directly that $\Theta\ge \pi$ (observe
that a non-self-intersecting arc of curvature pointwise
larger than $K$ has diameter $O(1/K)$). One 
then excludes the following two possibilities:
\begin{enumerate}
\item [{\bf (Case 1)}] $\Theta>\pi$ and $\gamma(t)$ does not shrink to a point as $t\to
  T^*$;
\item [{\bf (Case 2)}] $\Theta=\pi$.
\end{enumerate}
 
Case 1 can be ruled out following the argument of \cite[Theorem
4.1]{Grayson} and Case 2 is dealt with exactly like in  \cite[Section
5]{Grayson}. 
More precisely:
\begin{itemize}
\item whenever the author of \cite{Grayson} invokes his Theorem 2.1 (resp. Lemma 3.5), one should apply our
  Proposition
\ref{th:21} (resp. Proposition \ref{th:35}) instead;
\item Lemma 3.2 of \cite{Grayson} (the ``$\delta$-whisker Lemma'')
  holds also in our case (with the same proof), since it is based only
  on the maximum principle and on the fact that locally $\gamma(t)$
  evolves according to a strictly parabolic equation.
\end{itemize}

The only remaining possibility is that $\gamma(t)$ \emph{does shrink
to a point } when $t\to T^*$. Since the area of $\cD_t$ is
$2(T(\cD)-t)$ up to $T^*$ (Theorem \ref{th:esistenza}), we obtain that
$T^*=T(\cD)$,
that is the claim of Theorem \ref{th:gray}.

\medskip

As a side remark, for the \emph{isotropic} curve shortening flow Grayson
\cite{Grayson} proves also that $\gamma(t)$ becomes convex at a time
strictly smaller than the time when it shrinks to a point (so that $\Theta=2\pi$), and becomes asymptotically a circle
(of radius going to zero).

\subsection{Proof of Proposition \ref{th:21}}
\label{sub:21}

Recall the definition  $g^\go=a^\go k^\go$; in this section, for
lightness of notation, we drop the argument $\omega$ in
$g^\go,\gamma^\go(t)$, etc. We let $\gm(t)$ (resp. $g_{\rm min}(t)$) denote the maximum
(resp. minimum) of $g$ along $\gamma(t)$.
Let us first prove a weaker result, that is, that the curvature explodes on an arc of total curvature larger than 
$\pi/2$.
\begin{proposition}\label{ptit21}
Suppose that  $|g|$ is uniformly bounded by $K$ at time zero and that 
\[\gm(t)\ge A(K):=3K\exp(8K {\mathcal L}(\gamma(0))/\am)\]
 (respectively, $g_{\rm min}(t)\le -A(K)$) for some $t>0$. 
We recall
 that $\mathcal L(\gamma)$ is the length of $\gamma$.
Then there exists $u\le t$ and a sub-arc of $\gamma(u)$
with total curvature at least $\pi/2$ on which $g\ge K$ (respectively, $g\le -K$).
\end{proposition}
We will consider only the case $\gm(t)\ge A(K)$, the proof of the other statement
being essentially identical.
We are going to rely on the following Lemma,
whose proof is given at the end of this section: 
\begin{lemma}\label{zenob}
Let us call $R(t)=\cup_i R_i(t)$ the subset of $\gamma(t)$ on which
  $g\ge K$: it is a union of arcs of positive
  curvature and we assume that each $R_i(t)$ has an angle span
  $[a_i(t),b_i(t)]$ 
(the angle span can be larger than $2\pi$). 
The two following statements hold:
\begin{itemize}
\item [(i)]
If at time $t$ one has $\max_i (b_i(t)-a_i(t))\le \pi/2$  then
\begin{multline}
\label{ag}
 \frac{\dd}{\dd t}\int_{\gamma(t)} a\,\log \left(\frac gK\right)\ind_{g\ge K} \dd \theta:=\frac{\dd}{\dd t}\sum_i \int_{R_i(t)} a(\theta) \log\left(\frac{g}{K} \right) \dd
  \theta\le
  \\-(\gm(t)-K)_+^2-2K \frac{\dd}{\dd t}\mathcal L(\gamma(t)).
\end{multline}

\item [(ii)] If at time $t$ one has $\max_i (b_i(t)-a_i(t))\le \pi$  then 
\begin{equation}
\label{gag}
 \frac{\dd}{\dd t}\int_{\gamma(t)} a\,\log (g/K)\ind_{g\ge K} \dd \theta\le-2K \frac{\dd}{\dd t} \mathcal L(\gamma(t)) .
\end{equation}
\end{itemize}
\end{lemma}

\begin{proof}[Proof of Proposition \ref{ptit21}]

  Suppose that $\gm(t)\ge A:=A(K)$ and, by contradiction, that for all $u\le t$ there are no
  sub-arcs of length at least $ \pi/2$ on which $g\ge K$.  Integration of $(i)$
  of Lemma \ref{zenob} on $[0,t]$ gives
\begin{equation}\label{grojeu}
0\le \int_{\gamma(t)} a \log (g/K)\ind_{g\ge K} \dd \theta\le  -\int_0^t
(\gm(u)-K)_+^2 \dd u+2K\mathcal L(\gamma(0)).
\end{equation}
Here we have used the fact that
\begin{eqnarray}
  \label{eq:18}
  \int_{\gamma(0)} a \log
(g/K)\ind_{g\ge K} \dd \theta=0,
\end{eqnarray}
since at time zero $g$ is dominated by $K$ by assumption.


Moreover from the assumption $\gm(t)\ge A(K)$ and from \eqref{eq:dgmax}
one has
\begin{equation}
\gm(u)\ge \left[ A(K)^{-2}+\frac{2}{\am}(t-u) \right]^{-1/2}
\end{equation}
for any $u\le t$.
Hence when $(t-u)\le \am/(16K^2)$, we have $\gm(u)\ge 2K$ 
(i.e. $\gm(u)-K\ge \gm(u)/2$) and thus 
\begin{multline}
\int_0^t (\gm(u)-K)_+^2 \dd u
\ge \frac{1}{4}\int_0^{\am/(16K^2)} \gm(t-u)^2\dd u \\
 \ge  \frac{1}{4} \int_0^{\am/(16K^2)} \left[ A^{-2}+\frac{2}{\am}u\right]^{-1}\dd u\\
=\frac{\am}{8}\log\left(1+\frac{A^2}{8K^2}\right)>\frac{\am}{4}\log (A/(3K))=2K \mathcal L(\gamma(0))
\end{multline}
which contradicts  \eqref{grojeu}.
\end{proof}

\begin{proposition}\label{gro21}
  Suppose that $|g|$ is uniformly bounded by $K$ at time zero and that
  $\gm(t)\ge B(K):=A(K_1)$ (resp. $g_{\rm min}(t)\le -B(K)$), where
  $K_1=K\exp(2\mathcal L(\gamma(0))K/ \am )$ and $A(\cdot)$ is as in
  Proposition \ref{ptit21}.  Then there exists $u\le
  t$ and a sub-arc of $\gamma(u)$ with curvature at
  least $\pi$ on which $g\ge K$ (resp. $g\le -K$).
\end{proposition}

\begin{proof}
Assume to fix ideas that $\gm(t)\ge A(K_1)$.
From Proposition \ref{ptit21}, there exists a time $t_1\le t$ such that 
\begin{equation}
\label{eq:babau}
 \int_{\gamma(t_1)} a\,\log (g/K)\ind_{g\ge K} \dd \theta \ge \frac{\pi}{2} \am \log (K_1/K)\ge \pi  \mathcal L(\gamma(0)) K.
\end{equation}
If one assumes that until time $t_1$ there is no arc with curvature at least $ \pi$ on
which $g\ge K$, from Lemma \ref{zenob} $(ii)$, we have 
(using also \eqref{eq:18})
\begin{equation}
 \int_{\gamma(t_1)} a(\theta)\log (g/K)\ind_{g\ge K}\dd \theta\le 2 K \mathcal L(\gamma(0)),
\end{equation}
which  contradicts \eqref{eq:babau}.
\end{proof}

\begin{proof}[Proof of Proposition \ref{th:21}  (Conclusion)]
When $t$ approaches
$T^*$, the maximum of $|k|$ and therefore the maximum of
$|g|$ diverges (by definition of $T^*$). Therefore, for any $K$ there is $t<T^*$ such that 
either $\gm(t)>B(K)$  or $g_{\rm min}(t)\le -B(K)$ and Proposition \ref{gro21} allows to conclude.
\end{proof}
\smallskip

\begin{proof}[Proof of Lemma \ref{zenob}]
  One has
\begin{equation}
  \frac{\dd}{\dd t}\int_{R_i(t)} a(\theta)\log (g/K)\dd \theta  =\int_{R_i(t)} a(\theta)\partial_t(\log (g/K))\dd \theta,
\end{equation}
as on the (moving)  boundary of $R_i(t)$ one has $\log (g/K)=0$.
We know from \cite[Lemma 2.1]{GL2} that, on arcs where curvature has a constant
sign,
\begin{eqnarray}
  \label{eq:gthetatheta}
a\partial_t \log g=(g_{\theta\theta}+g)g.  
\end{eqnarray}
Then the left-hand side of \eqref{ag} is equal to
\begin{equation}
\sum_i\int_{R_i(t)}( g_{\theta\theta}+g)g\dd \theta= \sum_i\int_{R_i(t)} (-g_{\theta}^2+g^2)\dd \theta+K\sum_i [g_{\theta}(b_i)-g_\theta(a_i)].
\end{equation}
The last term is negative as $g_{\theta}(b_i)\le0$, $g_\theta(a_i)
\ge0$ by assumption.
We have to control the integral.
For this we use Wirtinger's inequality (which can be obtained by a Fourier decomposition of $f$ on a base of eigenfunction of the Laplacian)
for $(g-K)$ for each arc:
\begin{lemma}[Wirtinger's inequality] 
\label{lem:wirt}
Let $f$ be a $C^1$
  function on an interval $[a,b]$.
  If $f(a)=f(b)=0$ with $a\le b$, then 
\[
\int_a^b f^2(x)\dd x \le \frac{(b-a)^2}{\pi^2}\int_a^b\left(
\frac {\dd f}{\dd x}\right)^2\dd x
.
\]
\end{lemma}
Under assumption $(ii)$ we thus have for all $i$
\begin{equation}
 \int_{R_i(t)} (-g_{\theta}^2+(g-K)^2)\dd \theta\le  0.
 \end{equation}
 Therefore, 
 \begin{equation}
 \int_{R(t)} (-g_{\theta}^2+g^2)\dd \theta\le  \int_{R(t)} 2Kg\dd
 \theta\le 2K \int_{\gamma(t)} ak^2\dd s= -2K\frac{\dd}{\dd
   t}\mathcal L(\gamma(t))
 \end{equation}
 using $g=ak$ and $\partial_s\theta=k$ in the second inequality and \eqref{eq:dleng} in the last step.
 Under assumption $(i)$ we do the same thing, except for the arc where the maximal curvature is attained if it is larger than $K$.
 On this arc (call it $R_1(t)$) one has, using the fact that its total
 curvature is smaller than $\pi/2$
and Wirtinger's inequality,
 \begin{equation}
 \int_{R_1(t)} (-g_{\theta}^2+(g-K)^2)\dd \theta\le - \frac 3 4 \int_{R_1(t)} g_{\theta}^2\dd \theta \le -\frac{6}{\pi} (\gm-K)^2:
 \end{equation}
the last step just uses the fact that
if $|R_1(t)|$, the
angle span of $R_1(t)$, satisfies $|R_1(t)|\le \pi/2$, then
\[
\int_{R_1(t)} g_{\theta}^2\dd \theta\ge |R_1(t)|^{-1}\left(
\int_{R_1(t)}|g_\theta|d\theta\right)^2\ge \frac2\pi [2(\gm-K)]^2
\]
(apply Cauchy-Schwartz for the first inequality).
\end{proof}

\subsection{Proof of Proposition \ref{th:35}}
\label{sub:35}
To fix ideas we will assume that the curvature at the point $p$
mentioned in Proposition \ref{th:35} is
positive.
Decompose $\gamma(t)$ into \emph{minimal arcs}, i.e. arcs
where the curvature has constant sign, and which are delimited by inflection points
that evolve continuously (see again Lemma \ref{analytic}).
We want to find $\tilde K$ such that the curvature on minimal arcs of
total curvature smaller than $\pi$ is necessarily bounded by $\tilde K$.
By choosing $\tilde K$ large, we can restrict ourselves to an
arbitrarily small time neighborhood of $T^*$.

\smallskip

Recall from Lemma \ref{analytic} that the number of inflection points
of $\gamma(t)$ is decreasing with time  and that points of zero
 curvature that are not inflection points can %
 be present only for a finite
 set of times.
Hence we can, without loss of generality, consider a
time interval $(T^*-\varepsilon,T^*)$ where the number of inflection
points is constant while points of the latter type are absent.
  Note
that on the interval $(T^*-\varepsilon,T^*)$, by point $(iii)$ of Lemma
\ref{analytic}, the total curvature of minimal arcs is strictly
decreasing.  
As the number of such arcs is
finite, we can suppose (one might need to take $\gep$ smaller still)
that after time $T^*-\gep$, on all such arcs, the total curvature is
smaller than $\pi-2\gep$. 
We emphasize that $\gep>0$ depends only (but
in a very implicitly way) on
the initial condition $\gamma(0)$. In the rest of the proof one shows that the
curvature remains bounded on each one of these arcs. We look at one of
them, that we call $C(t)$.

Let us call $C^1(t)$, $C^2(t)$ the two minimal arcs of negative
curvature that are connected to the endpoints of $C(t)$. 
 It may happen that $C^1(t)=C^2(t)$, when $\gamma(t)$ has only two
 inflection points.
By definition of $T^*$, there exists a $K$ such that $T^*-\gep\le T_K$. Hence we can assume that 
at time $T^*-\gep$ the curvature of $\gamma$ and of all the curves $\gamma^\go$ for $\go$ sufficiently small are bounded above by $K$.

For $K'>K$ fixed, for $t\in(T^*-\gep,T_{K'}]$, when $\go$ is small
enough, we want to approximate $C(t)$ by an arc $C_\go(t)$ of
$\gamma^\go(t)$.  What we do first is finding minimal arcs $C^1_\go(t)$
and $C^2_{\go}(t)$ of $\gamma^\go(t)$, of negative curvature, approximating $C^1(t)$ and
$C^2(t)$.  At the midpoint $M_1(t)$ of $C_1(t)$, the curvature is
negative and bounded away from zero in
 $(T^*-\gep,T_{K'}]$. Hence, from
Proposition \ref{th:continua}, if $\go$ is sufficiently small, all the points of
$\gamma^\go(t)$ that lie in the vicinity of $M_1(t)$ have negative curvature,
hence they belong to a common minimal arc $C^1_\go(t)$ which has negative curvature.
A similar construction gives  $C^2_\go(t)$.

There are two arcs that connect $C^1_\go(t)$ and $C^2_\go(t)$: we call 
$C_\go(t)$ the one whose distance from $C(t)$ as $\go\to0$ tends to zero.
 Note that this is not necessarily a minimal
arc. We claim however that $C_\go(t)$ converges in Hausdorff distance
to $C(t)$ as $\go\to0$, uniformly for
$t\in(T^*-\gep,T_{K'}]$.  
Indeed by Proposition \ref{th:continua} its endpoints $p^1_\go(t),p^2_\go(t)$, on which curvature is zero,
must converge to points of 
zero curvature and the only possible options are the inflection points $p^1(t),p^2(t)$ that delimit
 $C(t)$. 


\begin{figure}[hlt]
 \begin{center}
 \leavevmode 
 \epsfxsize =6  cm
 \psfragscanon
 \psfrag{C}{\small $C(t)$}
 \psfrag{C1}{\small $C^1(t)$}
  \psfrag{C2}{\small $C^2(t)$}
 \psfrag{Cgo}{\small $C^\go(t)$}
 \psfrag{Cgo1}{\small $C^\go_1(t)$}
  \psfrag{Cgo2}{\small $C^\go_2(t)$}
 \epsfbox{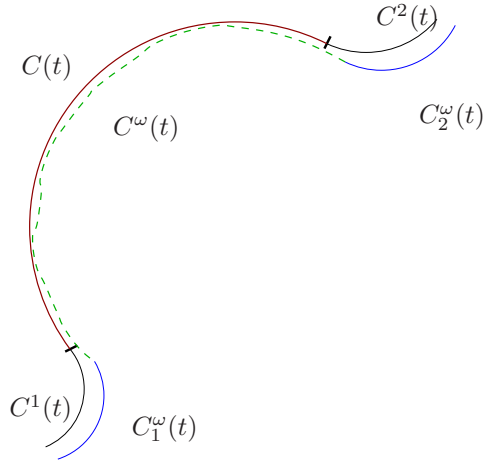}
 \end{center}
 \caption{ Scheme of approximation of $C(t)$ by an
 arc of $\gamma^\go(t)$.}
\label{piedestal}
 \end{figure}


\begin{lemma}
Fix $K'>K$.
\begin{itemize}
 \item [(i)] For $\go$ small enough, the integral of $|k^\go|$ along $
   C_\go(t)$ is bounded by $\pi-\gep$ for $t\in(T^*-\gep,T_{K'}]$.
 
 \item [(ii)] For all $t\le T_{K'}$ the curvature on $C_\go(t)$ is
   uniformly smaller than $ K (3/\gep)^m$ where $2m$ is the number
   of changes of sign of the curvature along $\gamma(0)$.
\end{itemize}
\label{lem:dpr}
\end{lemma}

\begin{proof}
 For point $(i)$, remark that we have
$$\int_{C_{\go}(t)} |k^\go|\dd s= \int_{C_{\go}(t)} k^\go \dd s  +2\int_{C_{\go}(t)}k^\go_- \dd s,$$
where $k^\go_-=\max(0,-k^\go)$.  We first show that the second
integral in the r.h.s., i.e.\ the total curvature $C_\go(t)$ restricted to sub-arcs
where the curvature is negative, is smaller than $\gep/4$.  As
$C_\go(t)$ converges to $C(t)$ and since curvature also converges, the
total length of arcs with negative curvature in $C_\go(t)$ is
vanishing. The curvature is bounded (by $K'$) and this  is sufficient to
conclude.

The first integral is the angle between the tangents at the endpoints
of $C_\go(t)$.  Because the curvature is bounded (by $K'$) the convergence of
$C_\go(t)$ to $C(t)$ implies that the tangent directions at the
extremities also converge.  Hence $\int_{C_{\go}(t)} k^\go \dd s $
converges to the total curvature of $C(t)$ which is smaller than
$\pi-2\gep$. Therefore, the first integral is smaller than $\pi-(3/2)\gep$ for
$\go$ small enough.
 
 \medskip
 
 For point $(ii)$, we use the ideas from \cite[Lemma 3.7]{Grayson}.
 We observe that $C_\go(t)$ is in
general composed of adjacent minimal arcs $c_i(t),1\le i\le 2n+1,$ for
some $n=n(t)\ge0$, with curvatures of alternating sign and $c_1(t)$
having positive curvature. The number $n(t)$ is decreasing in time, as
arcs disappear when two inflection points merge, and is in any case upper bounded by $m$, the number of positive curvature arcs on $\gamma(0)$.

 \begin{lemma}\label{lem37}
Let us consider an arc  $c(t)\subset \gamma^\go(t)$ where the curvature is
non-negative and whose endpoints are inflection points. Assume that,
at some time $s$,
\begin{itemize}
\item [(a)] $g$ is bounded by $K_1$ on $c(s)$;
\item [(b)] the total curvature of $c(s)$ is bounded by $\pi-\gep$.
\end{itemize}
Then until the first time $t_1$ when an inflection point at one of the extremities of $c(t)$ disappears, $g\le K_1/\sin(\gep/2)$ on $c(t)$.
\end{lemma}


We apply Lemma
\ref{lem37} for each arc  $c_{2i+1}(t)$ of positive curvature,
starting at time $s=T^*-\gep$, and
iterate it when two such arcs merge.  Since the number of such mergings cannot
exceed $m$, we obtain that $g^\go$ is bounded
above by $K/\sin(\gep/2)^m\le K(3/\gep)^m$. Recall that
$\gep$ is positive and depends only on the initial condition $\gamma(0)$.
\end{proof}
We conclude the proof of Proposition \ref{th:35} by noting that, when
$\go$ tends to zero, the curvature on $ C(t)$ is well approximated
by the one of $C_\go(t)$ (Proposition \ref{th:continua}), and that
$K'$ in Lemma \ref{lem:dpr} is arbitrary.
 Therefore, the curvature on
$C(t)$ is pointwise bounded by say $\tilde K:=2 K (3/\gep)^m$ for
$t\in(T^*-\gep,T^*)$ and we are done; we recall that $\gep>0$ is a
positive constant that depends (very implicitly) on $\gamma(0)$.

\medskip

\begin{proof} [Proof of Lemma \ref{lem37}]

 Recall that an arc of positive curvature can be parametrized by the tangent direction $\theta$ and (cf. \eqref{eq:gthetatheta})
 that $g^{\go}(\theta,t):=a^{\go}(\theta) k^\go(\theta,t)$ satisfies 
 (cf. \cite[Lemma 2.1]{GL2})
 \begin{equation}\label{croco}
 \partial_t g^\go=\frac{1}{a^\go}(g^\go)^2(g^\go_{\theta\theta}+g^\go).
 \end{equation}
 
 With no loss of generality we can consider that the initial angle span $I$
of $c$ (at time $s$) is included in $[\gep/2,\pi-\gep/2]$ and from Lemma 
\ref{unpointneuf} (ii) we know
that this remains true up to $t_1$.
Then we remark that initially for all $\theta$ in $I$ 
\begin{equation}
g^\go(\theta,0)\le \frac{K_1}{\sin(\gep/2)}\sin(\theta),
\end{equation}
and that the r.h.s.\ is a stationary solution of \eqref{croco}.  This,
via the maximum principle, implies that $g^\go(\theta,t)\le
\frac{K_1}{\sin(\gep/2)}\sin(\theta)$ for all $t\le t_1$.
\end{proof}

\section{Scaling limit of the droplet evolution: proof of Theorem \ref{mainres}}
\label{sec:mainres}
\subsection{Monotonicity}
\label{monoton}
We will use at various places (often implicitly) the well-known monotonicity or
attractiveness of the
dynamics, that we can formulate as follows. 
Consider the dynamics in a
domain  $V\subset (\bbZ/L)^2$, with boundary conditions $\tau$ on $\partial
V=\{x\in(\bbZ/L)^2\setminus V: d(x,V)=1/L\}$. In general, $\tau$
can
be time-dependent, $\tau=(\tau_x(t))_{x\in\partial V, t\ge0}$ . Let $\sigma^{\eta;\tau}(t)$ denote the
configuration at time $t$, when the initial condition is $\eta$ and
the boundary conditions are $\tau$. Also introduce in the space of
spin configurations the partial order $\preceq$ where $\sigma\preceq\sigma'$
if $\sigma_x\le \sigma'_x$ for every $x$.
Then, it is possible to define a
global coupling $\mathbb P$ such that, if $\eta\preceq \eta'$ and $\tau\preceq \tau'$
one has
\[
\sigma^{\eta;\tau}(t)\preceq \sigma^{\eta';\tau'}(t)\text{ for every\;} t\ge0
\]
with $\mathbb P$-probability $1$. Here, $\tau\preceq \tau'$ means
$\tau(t)\preceq \tau'(t)$ for every $ t\ge0$.

\subsection{Proof of Theorem \ref{mainres}}

The way to prove the result is slightly indirect.
We first show that if the initial droplet includes a strict neighborhood of $\mathcal D$
then $\cD_t$ is w.h.p. a lower bound for the droplet $\mal(L^2t)$ (we
also prove
an analogous upper bound). 

\begin{proposition}\label{controlinterieurexterieur}
Let $\cD$ be a compact set such that $\gamma(0)=\partial \cD$
satisfies the assumptions of Theorem \ref{th:existence}. 
\begin{enumerate}
\item If for some $\eta>0$
 the initial condition of the stochastic dynamics satisfies 
$
\mathcal A_L(0) \supset \mathcal D^{(\eta)}
$
then, w.h.p., 
\begin{equation}
\mathcal A_L(L^2t) \supset  \mathcal D_t \text{\quad for all\quad}
0\le t\le T(\cD)-\eta;
\end{equation}
\item If instead 
$
\mathcal A_L(0) \subset \mathcal D^{(-\eta)}
$
then, w.h.p., 
\begin{equation}
\mathcal A_L(L^2t) \subset  \mathcal D_t \text{\quad for all\quad} 0\le t\le T(\cD^{(-\eta)}).
\end{equation}
\end{enumerate}
\end{proposition}

\begin{proof}[Proof of Theorem \ref{mainres} given Proposition \ref{controlinterieurexterieur}]\

  \emph{Lower bound}. We prove that $\mal(L^2t)\supset
  \mathcal D_t^{(-\eta)}$ for $0<t\le \tilde T$, 
 where $\tilde T$  is the time when $\cD_t^{(-\eta)}$ becomes empty.  Take $\epsilon>0$
  and let $\cD_t^-$ be the deterministic flow started from initial
  condition $\cD^{(-\epsilon)}$ and $\gamma^-(t):=\partial
  \cD_t^-$. 
We assume $\epsilon$ is small enough (cf. Remark \ref{rem:bij}) so that 
the boundary of $\cD^{(-\epsilon)}$ is a simple curve and satisfies the assumptions of Theorem  \ref{th:existence}.
Proposition \ref{controlinterieurexterieur} gives
$\mal(L^2t)\supset \cD^-_t$
for all times up to $
T(\cD^{(-\epsilon)})-\epsilon$,
which is larger than $\tilde T$ for
$\epsilon$ small enough.  
The domain $\cD^-_t$ is included
in $\cD_t$ for all times.
We want to show that $\cD^-_t\supset \cD^{(-\eta)}_t$. Let us call $d(t)=\sup\{d(p,\gamma(t)):p\in
\gamma^-(t)\}$ with $d(p,\gamma)$ the distance from $p$ to $\gamma$.
Assume for definiteness that there is a unique pair of points ($p_t\in
\gamma^-(t), q_t\in \gamma(t)$) that is exactly at distance $d(t)$
(the general case is analogous): then necessarily the slopes
$\theta_t$ of $\gamma^-(t),\gamma(t)$ at these points are equal and
the segment $[p_t\,q_t]$ is normal to both $\gamma^-(t)$ and
$\gamma(t)$. Moreover, since $\gamma(t),\gamma^-(t)$ are smooth curves, we can locally
approximate them by arcs of circles of radii $1/k
(q_t)$, $1/k^-(p_t)$ respectively, with $k(q_t)$ the curvature of
$\gamma(t)$  at $q_t$  and $k^-(p_t)$ the curvature of $\gamma^-(t)$ at
 $p_t$.  The fact that $q_t$ realizes the infimum of $d(p_t,q)$
for $q$ ranging on $\gamma_t$ implies that $k(q_t)d(t)\le 1$. On the
other hand, the fact that $p_t$ maximizes $d(p,\gamma(t))$ for $p$
ranging on $\gamma^-(t)$ implies
\begin{equation}
\label{soff1}
k^-(p_t)\leq \frac{k(q_t)}{1-k(q_t)d(t)}.
\end{equation}
Both inequalities are easy to check if $\gamma(t), \gamma^-(t)$ are
replaced by arcs of circles.

A look at \eqref{eq:mc} then shows that, as long as $2d(t)k_{\max}<1$,
\begin{equation}
\label{soffr}
\frac{\dd}{\dd t}d(t)= 
a(\theta_t)(-k(q_t)+k^-(p_t))\le 
2a_{\max}k^2_{\max}d(t)
\end{equation}
where $k_{\max}<\infty$ is the maximal curvature of $\gamma(t)$ up to time 
 $\tilde T$. 
Observe that  $d(0)=O(\epsilon)$. Then, choosing $\epsilon$ sufficiently
small (as a function of $\eta$), \eqref{soffr} ensures that $d(t)$ remains
smaller than $\min(1/(2k_{\max}),\eta/2)$ up to $\tilde T$.
As a consequence,  $\cD^-_t\supset
\cD_t^{(-\eta)}$. The inclusion $\mal(L^2t)\supset \mathcal
D_t^{(-\eta)}$ is therefore proven up to time $\tilde T$, as we wished.

\medskip

\emph{Upper bound}.
Let $\cD_t^+$ be the deterministic flow started from initial
  condition $\cD^{(\epsilon)}$ and fix $\xi$ small. It follows from Proposition \ref{controlinterieurexterieur} that, until time
  $T(\cD)$,
  $$\mal(L^2t)\subset \cD^+_t.$$ By the same argument that showed that
  $\cD^-_t\supset \cD_t^{(-\eta)}$ we have that $\cD^+_t\subset \cD_t^{(\eta/4)}$
  until time $T-\xi$, provided $\epsilon$ is small enough.
This proves the upper bound up to time $T(\cD)-\xi$. 

 Then  we notice that,  if $\xi$ has been chosen small enough,
$\cD_{T-\xi}^{(\eta/4)}$ is included in $ B(X,\eta/2)$, a ball of
radius $\eta/2$ centered at $X$, 
the point to which $\gamma(t)$
shrinks as $t\to T(\cD)$.
 From \cite[Theorem 2.2]{LST}
 we know that, w.h.p., an initial droplet contained in $ B(X,\eta/2)$
 disappears within time $O(L^2 \eta^2)$, and  is 
included at all times in $B(X,2\eta/3)$. On the other hand, always for
$\xi$ small, $B(X,2\eta/3)\subset \cD_t^{(\eta)}$ for all $t\ge T-\xi$. This
concludes the proof of both claims.
\end{proof}


\subsection{Proof of Proposition \ref{controlinterieurexterieur}}

To prove Proposition \ref{controlinterieurexterieur} one needs two ingredients.
The first says essentially that in the diffusive scaling the speed of evolution of the
random droplet boundary is finite: in a time $L^2\gep$, it moves at
most a distance $O(\gep)$.

\begin{proposition}\label{chips}
Suppose that $\cD$ satisfies  the assumptions of Theorem
\ref{th:existence}, so that in particular $m=m(\cD)>0$ (cf. definition
\eqref{eq:13}), 
and
consider the dynamics starting from initial condition \eqref{bien}.
There exists $\gep_0(m)>0$ and $C_1(m)<\infty$ such that for any $\gep\in(0,\gep_0)$, with high probability  
\begin{equation}
\label{eq:chips}
\mathcal D^{(-C_1 \gep)}\subset   \mal( L^2 t)\subset  \mathcal
D^{(C_1\gep)}\quad
\text{for every}\quad t \in[0,\gep].  
\end{equation}
The constants $C_1$ and $1/\varepsilon_0$ can be chosen to be decreasing in $m$.
\end{proposition}

The second ingredient is a control of the droplet after a time
 small
on the diffusive scale $L^2$.

\begin{proposition}\label{infinit}
Let $\cD$ satisfy the assumptions of Theorem
\ref{th:existence} and fix $c>0$. There exists $\l_0(c,\cD)>0$ and,
for $0<\l<\l_0$, there exists $\gep_1=\gep_1(\lambda,c,\cD)>0$ such that 
for all for $\gep\le \gep_1$, for all integer $j$ with
$j\gep\le T(\cD)-c$,
\begin{enumerate}
\item If the initial condition (which might be random) satisfies w.h.p. 
$
\mathcal A_L(0) \supset \mathcal D_{j\gep}^{(\lambda (T-j\gep))}
$
then w.h.p.
\begin{equation}
\mathcal A_L(L^2\gep) \supset \mathcal D_{(j+1)\gep}^{(\lambda (T-(j+1)\gep))}.
\end{equation}

\item If the initial condition satisfies w.h.p.
$\mathcal A_L(0) \subset \mathcal D^{(-\lambda(T-j\gep) )}$
then w.h.p.
\begin{equation}
\mathcal A_L(L^2\gep) \subset \mathcal D_{(j+1)\gep}^{(-\lambda (T-(j+1)\gep))}.
\end{equation}
\end{enumerate}
\end{proposition}
Recall that $\DD_{j\gep}^{(\lambda (T-j\gep))}$ is the deterministic
domain $\cD$ at time $j\gep$, expanded by $\lambda(T-j\gep)$.

\begin{proof}[Proof of Proposition \ref{controlinterieurexterieur} given Propositions \ref{chips} and \ref{infinit}]\
  Remark that we can assume that $\eta$ is small: indeed, if the claim of Proposition \ref{controlinterieurexterieur}
holds for small values of $\eta$, then it clearly holds also for larger values.

\emph{Claim (1)}.
Choose $c=\eta$, $\gl = \min(\eta/T(\cD),\lambda_0(c,\cD))$ and $\gep$
smaller than $\gep_1(\lambda,c,\cD)$, with $\lambda_0$ and $\gep_1$
from Proposition \ref{infinit}. Iterating the first claim of Proposition \ref{infinit}, we obtain
  w.h.p.\ for all
 $j$ such that $j\gep\le T-c$ 
 \begin{equation}
  \mathcal A_L(L^2 j\gep) \supset \mathcal D_{j\gep}^{(\lambda
    (T-j\gep))}\supset \cD_{j\gep}.
 \end{equation}
This is the desired statement, but only at the discrete set of times
$j\gep L^2$.

 Then, in order to also have an inclusion bound in the time intervals
 $[j\gep,(j+1)\gep]$, we use Proposition \ref{chips}. 
First of all, define $m_{\min}=\min_{t\le T-c}m(\mathcal
 D_{t}^{(\lambda (T-t))})$ 
 (recall \eqref{eq:13}). One has $m_{\min}>0$. If $\gl$ were zero,
 this would follow immediately  from Lemma \ref{ouille}, since the
 inverse maximal curvature $r(t)$ is bounded away from zero up to time
 $T(\cD)-c$. For $\lambda $ small (i.e. $\eta$ small), instead, $m_{\min}>0$ follows from Remark \ref{rem:bij} and Lemma \ref{ouille}, which relate explicitly the curvature
function of $\cD_t $ with that of
$\cD_t^{(x)}$ for small $x$. 

Assume then that $\gep$ was chosen smaller than $\gep_0(m_{\min})$,
with $\gep_0$ as in Proposition \ref{chips}. 
Starting at time $j\gep$ with an
 initial condition ``$-$'' in $\mathcal D_{j\gep}^{(\lambda (T-j\gep))}$
 and ``$+$'' outside, we get that w.h.p. 
 \begin{equation}
  \mathcal A_L(L^2 t) \supset \mathcal D_{\gep \lfloor t/\gep \rfloor}^{(\lambda (T-\gep \lfloor t/\gep \rfloor)-C_1\gep)}
 \end{equation}
 for all $t\in[j\gep,(j+1)\gep]$, and therefore (repeating the
 argument for all values of $j$) for all $t\le T-c$.  The set on the
 right hand side contains $\mathcal D_{t}^{(\lambda (T-\gep \lfloor
   t/\gep \rfloor)-C_2\gep)}$,
   provided that $C_2\ge C_1+k_{\max}$ (we used that $a(\cdot)<1$).
In turn,
 $\mathcal D_{t}^{(\lambda (T-\gep \lfloor t/\gep
   \rfloor)-C_2\gep)}\supset \cD_t$ for all $t\le T-c$ if $\gep$ is
 sufficiently small.

\medskip

\emph{Claim (2)}. This is proven analogously, taking $c\le
T(\cD)-T(\mathcal D^{(-\eta)})$.
\end{proof}

\subsection{Proof of Proposition \ref{chips}}

We are going to use a simple consequence of \cite[Theorem
2.2]{LST},  that gives the convergence of the stochastic evolution
$\mal(L^2t) $ to
the curve-shortening flow \eqref{eq:mc} in the case of a smooth convex
initial droplet.
\begin{lemma}\label{ptichips}
Set $r>0$, $x\in \bbR^2$ and $\epsilon>0$.
Consider the zero-temperature stochastic Ising model starting from an initial condition which satisfies 
$$\cA_L(0)\supset \cB (x, r).$$
Then w.h.p. for all $t\in [0,\epsilon]$
$$\cA_L(L^2 t)\supset \cB(x,\sqrt{ r^2-4 a_{\max} \epsilon})$$
where $\cB(x,r)$ denotes the open ball of center $x$ and radius $r$ if $r>0$ and the empty set otherwise.
\end{lemma}

\begin{proof}
According to \cite[Theorem 2.2]{LST}, 
for any positive $\eta$, w.h.p. for all $t\leq \epsilon$
$$\cA_L(L^2 t)\supset (\xi^{r,x}_t)^{(-\eta)}.$$
where $(\xi^{r,x}_t)_{t\geq 0}$ is the solution of \eqref{eq:mc} with initial condition $\cB(x, r)$.
The boundary of the  ball $(\cB(x,\sqrt{ r^2-4 a_{\max} t}))_{t\geq
  0}$ is instead solution of 
the \emph{isotropic} curve-shortening flow
\begin{equation}
\label{eq:mcmodif}
\partial_t \gamma:=  2 a_{\max} k {\bf N}.
\end{equation}
The domain $\xi^{r,x}_t$ is convex at all times \cite{LST}, so the
velocity always points inward. Since $a(\theta)$ is strictly smaller
than $2a_{\max}$, equation \eqref{eq:mcmodif} shrinks convex domains
strictly faster than the original equation \eqref{eq:mc}. Therefore,
one can find $\eta$ (depending on $\epsilon$) such that for all $t\le
\epsilon$,
$$ (\xi^{r,x}_t)^{(-\eta)}\supset \cB(x,\sqrt{ r^2-4 a_{\max} \epsilon}).$$
\end{proof}

%
\begin{proof}[Proof of Proposition \ref{chips}]
 We first prove the lower bound of \eqref{eq:chips}.  Define $r:=
 m(\cD)/2$. 
Recalling Remark \ref{rem:bij}, note that for every $0<a<b\le r$ one has 
\begin{eqnarray}
  \label{eq:14}
  \cD^{(-(b-a))}=\cup_{x\in\cD^{(-b)}}\cB(x,a).
\end{eqnarray}
Choose $\gep_0(r)$ sufficiently small so that any $\gep<\gep_0(r)$ satisfies
 $$ 
 \sqrt{ r^2-4 a_{\max} \gep}\ge r- \frac{3 a_{\max}\gep}{r}.
 $$ 
In practice, think of $\gep_0(r)\ll r^2$.
Set $C_1:= 4 a_{\max}/r$.
We have 
\begin{equation}
\mathcal{D}^{(-C_1\gep)} \subset \bigcup_{x\in\mathcal{D}^{(-r)}} \cB(x,\sqrt{ r^2-4 a_{\max} \gep})
\end{equation}
from \eqref{eq:14}, that is applicable since 
\[
C_1\gep\le 4\frac{a_{\rm max}}r\gep_0(r)\ll r.
\]
As $\mathcal{D}^{(-C_1\gep)}$ is a compact set, one can extract a finite subset $J\subset  \mathcal{D}^{(-r)}$ which satisfies 
\begin{equation}\label{extraction}
\mathcal{D}^{(-C_1\gep)} \subset \bigcup_{x\in J} \cB(x,\sqrt{ r^2-4 a_{\max} \gep}).
\end{equation}
For each $x\in J$ we have $\cA_L(0)\supset \cB (x, r)$, thus we can apply Lemma \ref{ptichips} and a union bound to obtain that w.h.p.\
for all $t\le \gep$
\begin{equation}
\mathcal A_L(L^2t)\supset \bigcup_{x\in J} \cB(x,\sqrt{ r^2-4 a_{\max} \gep})\supset \mathcal{D}^{(-C_1\gep)} .
\end{equation}

The proof of the other inclusion of \eqref{eq:chips}, $\mal(L^2
t)\subset \cD^{(C_1\gep)}$, is similar since the roles of ``$+$'' and
``$-$'' spins are symmetric; we just have to take care to work with
compact sets while the set of ``$+$'' spins is not compact.  For
$\gep<\gep_0(r)$, let $\mathcal{R}$ be a rectangle, with sides
parallel to the coordinate axes, that contains
$\mathcal{D}^{(C_1\gep)}$.  From the definition of the dynamics, the
spins outside $\mathcal R$ remain ``$+$'' at all times, since they
always have
 four ``$+$'' neighbors.  Let $U$ be the
closure of $\left(\mathcal R \setminus \mathcal{D}^{(C_1\gep)}\right)$
($U$ is compact).  We have to prove that all spins in $U$ stay ``$+$ up
to time $\gep$.  Similarly to \eqref{extraction} one can find a finite
subset $J '\subset \mathcal R \setminus \mathcal{D}^{(r)}$ which
satisfies
\begin{equation}
U \subset \bigcup_{x\in J'} \cB(x,\sqrt{ r^2-4 a_{\max} \gep}),
\end{equation}
and then the conclusion follows from Lemma \ref{ptichips} applied to the dynamics where the roles of ``$+$'' and ``$-$'' are reversed.
%
%
 \end{proof}

\subsection{Proof of Proposition \ref{infinit}}
We give full details only for the proof of the first claim, the proof of second
one being very similar.
For notational convenience we prove the result for $j=0$, and explain briefly
in Remark \ref{rem:genj} why the proof remains valid for all $j$ such that $j\gep\le T-c$.

Given $\gl>0$ we set $\mathcal{G}_t:=\cD^{(\gl(T-t))}_t$. Let
$\chi(t)$ denote the boundary of $\mathcal{G}_t$ and $\kappa$ the
curvature function associated to it.  By monotonicity, it is
sufficient to prove that if one starts with the initial condition
``$-$'' in $\mathcal{G}_0$ and ``$+$'' outside, then for $\gep$
smaller than $\gep_1(\lambda,c,\DD)$, w.h.p.
\begin{equation}\label{inclaprov}
\mathcal A_L(L^2\gep) \supset \mathcal{G}_\gep.
\end{equation}
The parameter $\l$ will be chosen small in Lemma \ref{voldechi}.
We can use Proposition \ref{chips} which says that w.h.p.\ $\mathcal{G}_0^{(-C_1\gep)}$ is filled with ``$-$'' spins at time $\gep$.
Hence we just have  to prove that w.h.p.
\begin{equation}\label{inclv}
\mathcal A_L(L^2\gep) \supset  \mathcal{G}_\gep \setminus   \mathcal{G}_0^{(-C_1\gep)}.
\end{equation}
We denote by $V_\gep$ the compact closure of  $\mathcal{G}_\gep
\setminus   \mathcal{G}_0^{(-C_1\gep)}$. Roughly speaking,
$V_\gep$ is the collection of points in  $\mathcal{G}_\gep$ that
are at distance 
$ C_1\gep$ from the boundary of $\cG_0$.

\smallskip

\emph{{\bf Sketch of the strategy.} The main idea to prove \eqref{inclv} is to control the motion of the
boundary of $ \mal(L^2t)$ around a given point $x\in\chi(0)$ via a
comparison with $\mal^x(L^2 t)$: this is the stochastic droplet  starting from a
circular shape
$ \cP^x$ which is tangent to $\chi(0)$ at $x$ and whose curvature is
close to that 
of $\chi(0)$ at $x$. Roughly speaking, $ \cP^x$ will sit in the
interior of $\chi(0)$ when the curvature at $x$ is positive and outside it
when the curvature is negative (some care will be needed when the
curvature is close to zero). Given that the two initial
curves $\chi(0)$, $\partial \cP^x$ have the same slope and almost 
the same curvature at $x$, at time zero and close to $x$ the
boundaries of $\mal(L^2 t)$ and $\mal^x(L^2 t)$ 
feel almost the same drift. Our work will then consist in proving that,
 locally around $x$ and   for small positive times,  they do
remain close.
On the other hand, $ \cP^x$ being convex, the evolution of
$\mal^x(L^2t)$ in the
scaling limit for every $t\ge0$ is precisely
controlled by Theorem 2.2 of  \cite{LST}. }

\smallskip

As we will see, in practice we need $\cP^x$ to be slightly more curved (in the ${\bf
  N}$ direction) than $\chi(0)$: we define therefore $\cP^x$ to be the disk tangent
to $\chi(0)$ and whose curvature vector at $x$ is equal to
$\bar\kappa(x){\bf N}$, with $\bar\kappa$ defined by
\begin{equation}
\label{eq:kprime}
 \bar\kappa(x):=\begin{cases}
              \kappa(x)+\lambda/100\quad\text{ if } \quad|\kappa(x)+\lambda/100|\ge \lambda/200\\
              \lambda/200 \quad \quad\quad\quad\,\text{ if } \quad|\kappa(x)+\lambda/100 |<\lambda/200.
             \end{cases}
\end{equation}
The second condition is there to prevent the curvature radius of $\cP^x$
from diverging.  Note that with this definition $\kappa$ and $\bar\kappa$
do not necessarily have the same sign but we have in any case
\begin{equation}\label{lesbornes}
\lambda/100  \le \bar\kappa(x)-\kappa(x)\le \gl/50.
\end{equation}
We set 
$$\bar \cP^x:=\begin{cases} \cP^x \,\quad\quad\quad\text{ if }\quad \quad \bar\kappa(x)>0, \\
             (\cP^x)^c \quad \quad\text{ if } \quad \quad\bar\kappa(x)<0
            \end{cases}$$
and we note that $\bar \cP^x$ is unbounded if $\bar\kappa(x)<0$.
Define then $\sigma^x(t)$ to be the dynamics starting from ``$-$''
in $\bar \cP^x$ and ``$+$'' in $(\bar \cP^x)^c$ at time zero, and
$\mal^x(t)$ as in \eqref{mal} with $\sigma(t)$ replaced by $\sigma^x(t)$.
As it is not true in general that $\bar \cP^x\subset \mathcal{G}_0$, there is no immediate comparison between $\mal^x(t)$
and $\mal(t)$.
However, the comparison works if we restrict ourselves to a
sufficiently small box around $x$. This is what we do next.

\smallskip

Consider
${\bf B}_x$ to be a square of side-length $2d>0$, centered at $x\in \chi(0)$,
whose sides are parallel to the normal and tangent vectors to
$\chi(0)$ at $x$.
If $d$ is chosen small enough (depending on $m(\cD)$, $\l$, the sup norm and the modulus of continuity of  $\kappa$, but neither on $\gep$ nor $x$),
the intersection of $\chi(0)$ with ${\bf B}_x$ has only one connected component and 
\begin{equation}\label{danse}
\left( \mathcal{G}_0 \cap {\bf B}_x\right) \supset  \left( \bar \cP^x \cap {\bf B}_x\right)     .                
\end{equation}
Furthermore, if $\gep$ is small enough (depending on $\l,d$ and on the
constant $C_1$ of Proposition \ref{chips})
then 
\begin{equation}\label{boundary}
((\bar \cP^x)^{(C_1\gep)} \cap \partial {\bf B}_x) \subset (\mathcal{G}_0^{(-C_1\gep)} \cap \partial{\bf B}_x).
\end{equation}

\begin{figure}[hlt]
 \begin{center}
 \leavevmode 
 \epsfxsize =8  cm
 \psfragscanon
 \psfrag{x}{$x$}
\psfrag{ceps}{$C_1\epsilon$}
\psfrag{bx}{${\bf B}_x$}
\psfrag{N}{$\mathbf{N}$}
\psfrag{Pbar}{$\bar{\mathcal{P}}^x$}
\psfrag{g0}{$\mathcal{G}_0$}
\psfrag{chi0}{$\chi(0)$}
\psfrag{T}{${\bf T}$}
 \epsfbox{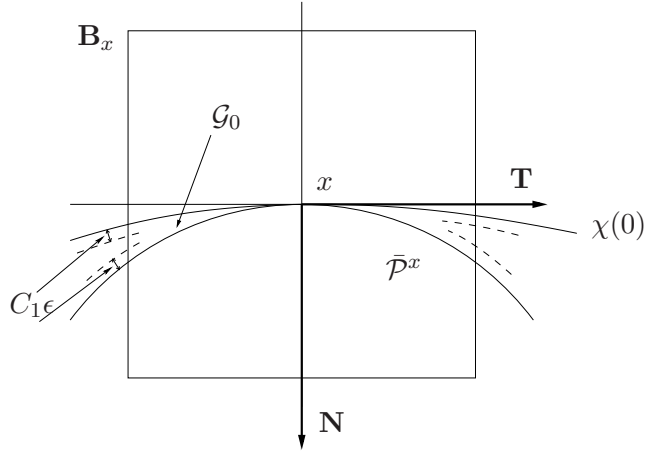}
 \end{center}
 \caption{Graphical representation of inclusions \eqref{danse} and \eqref{boundary}. We recall that $\chi(0)$ is the boundary of $\cG_0$.}
 \label{fig:65}
 \end{figure}

These properties, together with monotonicity of the dynamics (cf. Section
\ref{monoton}), allow us to establish a \emph{local} comparison,  that holds
w.h.p, 
between the two stochastic droplets $\mal(L^2t) $ and $ \mal^x(L^2t)$: 
\begin{lemma}
\label{lem:domsto}
If $d$ and $\gep$ are such that \eqref{danse} and \eqref{boundary} are
satisfied for the constant $C_1$ corresponding to Proposition \ref{chips} then for all $x$ in $\chi(0)$, w.h.p.
\begin{equation}
\label{eq:15}
\left( \mal^x(L^2 \gep)\cap {\bf B}_x \right) \subset \left( \mal( L^2\gep)\cap {\bf B}_x\right).
\end{equation}
\end{lemma}
Let $\cP^x_t$ denote the evolution of $\cP^x$ by the anisotropic curve
shortening flow \eqref{eq:mc} and
set 
\begin{equation}
\bar \cP^x_t=\begin{cases}\cP^x_t\,\quad\quad\quad \text{ if } \quad\quad\bar\kappa(x)>0,\\
                 (\cP^x_t)^c \quad\quad\text{ if } \quad\quad\bar\kappa(x)<0.
                \end{cases}
\end{equation}

The second important ingredient of the proof is to show: 
\begin{lemma}\label{inclus}
For $\gep$ small enough, there exists $\eta$ such that
\begin{equation}\label{cover}
V_\gep\subset \bigcup_{x\in \chi(0)} \left( (\bar \cP^x_\gep)^{(-\eta)}\cap \bB_x \right)
\end{equation}
where $V_\gep$ was defined just after \eqref{inclv}.
Furthermore the inclusion remains valid if the closed sets on the right-hand
side are replaced by their interiors.
\end{lemma}

\begin{proof}[Proof of Proposition \ref{infinit}, Claim (1), from Lemma \ref{lem:domsto} and \ref{inclus}]
As $V_\gep$ is a compact set, in the inclusion 
\eqref{cover} we can extract a finite set $J\subset \chi(0)$ which satisfies 
\begin{equation}\label{cover2}
V_\gep\subset \bigcup_{x\in J} \left( (\bar \cP^x_\gep)^{(-\eta)}\cap \bB_x \right).
\end{equation}
Then combining Lemma \ref{lem:domsto} and \cite[Theorem 2.2]{LST},
that guarantees that $\mal^x( L^2t)\supset (\bar \cP^x_t)^{(-\eta)}$, we notice that w.h.p.\ 
for every $x\in J$,
\begin{equation}
\mal( L^2\gep) \supset \left(\mal^x( L^2\gep)\cap {\bf B}_x\right) \supset \left( (\bar \cP^x_\gep)^{(-\eta)}\cap \bB_x \right).
\end{equation}
As a side remark, note that when $\bar\kappa(x)<0$ we need to reverse the role of ``$-$''
and ``$+$'' spins  when applying \cite[Theorem 2.2]{LST}: this is
because the bounded initial droplet $\cP^x$ is filled with ``$+$'' and not
``$-$'' spins.
Hence, as  $J$ is finite, we can take the union over $x\in J$ to get that w.h.p.
\begin{equation}
\mal( L^2\gep) \supset V_\gep.
\end{equation}
\end{proof}

\begin{proof}[Proof of Lemma \ref{lem:domsto}]
We observe the dynamics $\sigma^x(t)$ and $\sigma(t)$ restricted to
the square window ${\bf B}_x$. We let $\tau^x(t)$ and $\tau(t)$ denote the restriction of $\sigma^x(t)$
and $\sigma(t)$ to the boundary $\partial \bB_x$ of the box (i.e. the set of lattice sites
outside ${\bf B}_x$ that are at distance $1/L$ from some lattice site in ${\bf B}_x$).
From monotonicity of the dynamics,
to show \eqref{eq:15} it is sufficient to show that w.h.p.\
\begin{align}
&\left(\mal^x(0)\cap {\bf B}_x\right) \subset \left(\mal(0)\cap {\bf B}_x\right),\\
& \tau_y(s)\leq \tau^x_y(s) \quad\text{ for every }  \quad s\in [0,\gep], y\in \partial \bB_x.
\end{align}

The first point (domination between the initial conditions) is just an
immediate consequence of \eqref{danse}.  The second one (domination
between time-dependent boundary conditions) is a consequence of
Proposition \ref{chips} combined with \eqref{boundary}.
\end{proof}

\emph{Proof of Lemma \ref{inclus}.}
For all $x$ in $\chi(0)$, we define a local coordinate system $\mathcal{S}_x=(x, {\bf T},{ \bf N})$. 
Given $\delta$ (which will depend on $\gep$) we define $\mathbf M_x$
to be the rectangular box
$[-\delta,\delta]\times[-d,d]$ in the frame $\mathcal{S}_x$
(it is much narrower than the $2d\times 2d$ square $\bB_x$, i.e. we have to think $\delta\ll d$).
Note that given $d$, if $\gep$ is small enough, for all $\delta>0$ we have
\begin{equation}
V_\gep\subset \bigcup_{x\in \chi(0)} \mathbf M_x.
\end{equation}
Hence to prove \eqref{cover} it is sufficient to prove that for all $x\in \chi(0)$,
\begin{equation}
\label{eq:cle}
\mathcal{G}_\gep \cap \mathbf M_x  \subset \left( (\bar\cP^x_\gep)^{(-\eta)}\cap \mathbf M_x \right).
\end{equation}

We first rewrite this inclusion as an inequality between functions.
If $d$ and $\gep$ are sufficiently small (with $\gep$ that depends on
$d$), the restriction of $\chi(t)=\partial \cG_t$ and of $\partial
\bar\cP^x_t$ to ${\bf M}_x$ can be considered as the graphs of functions
in $\mathcal{S}_x$, for all $t\in[0,\gep]$.  We denote these functions
by $f^x(u,t)$ and $g^x(u,t)$ respectively, with
$u\in[-\delta,\delta]$. Recall that the reference frame $\cS_x$ is
such that the $y$-axis is directed along the \emph{inward pointing}
normal vector ${\bf N}$. The proof of \eqref{eq:cle} will thus be
complete if we can prove the following:
 \begin{lemma}
\label{lem:ratrape}
There exists $\gep_0(\cD,\gl)$ so that  for all $\gep<\gep_0$  there exists $\delta(\gep)$ and $\eta(\gep)$ such that for all $x\in \chi(0)$  
\begin{equation}\label{ratrape}
g^x(u,\gep)+\eta\le f^x(u,\gep), \qquad \forall u\in [-\delta,\delta].
\end{equation}
\end{lemma}
\begin{proof}[Proof of Lemma \ref{lem:ratrape}]
In order to simplify notations we drop the exponent $x$ in the following.
The core of the proof is to show that the drift of $\partial\cG_t$ in the ${\bf N}$ direction is stronger than the one of $\partial\bar\cP_t$.
Thus, even though $g$ starts above $f$ initially (cf. \eqref{danse}), it has time to catch
up. Note that, since $\partial\bar\cP$ is more curved than $\cG$ at
time zero (cf. \eqref{lesbornes}), it
would look like the drift of $\partial\bar\cP$ should be larger: however, we will see
that the boundary of $\cG_t$ solves the curve shortening flow
\eqref{eq:mc} with an extra term, proportional
to $\lambda$,   in the normal velocity. This extra term guarantees the desired inequality
between drifts.

Observe first of all that, provided that $\d$ is small enough, for all $u\in [-\d,\d]$
\begin{equation}
\begin{split}
g(u,0)&=\frac12\bar\kappa(x)u^2+O(u^4)\\
f(u,0)&\ge \frac12(\kappa(x)-\gl/50)u^2+O(u^4).
\end{split}
\end{equation}
 The $O(u^4)$ term is uniform in $x$ and just depends on the maximal curvature.
We have simply approximated locally the curves with suitable
parabolas, and  
the second inequality is valid provided that $\kappa>\kappa(x)-\gl/50$ on $\chi(0)\cap \mathbf M_x$ (this is true if $\d$ is small).
Hence
\begin{equation}\label{grimph}
f(u,0)\ge g(u,0)+\frac{1}{2}(\kappa(x)-\gl/50-\bar\kappa(x))u^2+O(u^4)\ge g(u,0)-\frac\gl{40} u^2,
\end{equation}
where the second inequality holds if $\d$ is chosen small enough.

What we want to show is that for all $u\in [-\delta,\delta]$, $t\in[0,\gep]$
\begin{equation}\label{groumph}
\partial_t f(u,t)- \partial_t g(u,t)\ge \gl/10.
\end{equation}
Then the equation \eqref{ratrape} will easily be derived by integrating \eqref{groumph} starting from \eqref{grimph}, provided that
\begin{equation}
\frac{ \gl\gep}{10}\ge \eta+\frac{\lambda}{40} \gd^2
\end{equation}
(for instance on can take $\gd=\sqrt{\gep}$ and $\eta=3\gl\gep /40$).

 To have an estimate on the time derivatives we need to use the equation that are satisfied by $\chi(t)=\partial\cG_t$ and $\partial\bar\cP^x_t$ respectively.
\begin{lemma}\label{voldechi}
  The curve $\chi(t)$ is solution of the modified curve-shortening equation
\begin{equation}\label{motus}
 \partial_t \chi= \left[a(\theta) \kappa+ \lambda\left(1+\frac{a(\theta)(T-t)\kappa^2}{1-\lambda(T-t)\kappa}\right)\right]\bf N,\\
\end{equation}
where $\bf N$ is the normal vector oriented inside the curve and
$\theta,\kappa$ are the slope and curvature at the point in question.
\end{lemma}
On the other hand, $\partial \bar\cP^x_t$ solves the usual curve
shortening flow \eqref{eq:mc},
that does not depend on $\lambda$.
Using simple trigonometry we then obtain that 
\begin{align}
 \partial_t f(u,t)&= \frac{1}{\cos(\theta-\theta_x)} \left(a(\theta) \kappa+ \lambda\left(1+\frac{a(\theta)(T-t)\kappa^2}{1-\lambda(T-t)\kappa}\right)\right),\\
\partial_t g(u,t)&= \frac{1}{\cos(\bar\theta-\theta_x)}a(\bar\theta)\bar \kappa
\end{align}
where $\kappa$ (resp. $\bar \kappa$) denotes the curvature on $\chi$
(resp. $\partial\cP^x$), $\theta$ (resp. $\bar
\theta$) the tangent
angle on $\chi$ (resp. $\partial\cP^x$) at the point of coordinate
$u$, and $\theta_x$ the tangent angle to $\chi(0)$ at $x$.

\medskip

By continuity of curvature in space and time (cf. Proposition
\ref{th:continua}), if one takes $\gep$ and $\delta$ small enough,
$\theta$ can be assumed to be arbitrarily close to $\theta_x$, $\kappa$
arbitrarily close to $\kappa(x)$ and $\bar{\kappa}$ arbitrarily close
to $\bar\kappa(x)$. Hence one has for all $u\in[-\delta,\delta]$,
$t\in[0,\gep]$
\begin{align}
\partial_t f(u,t)&\ge  a(\theta_x) \kappa(x)+ 2\lambda/3,\\
 \partial_t g(u,t)&\le  a(\theta_x) \bar\kappa(x)+\lambda/3.
\end{align}
Then we conclude the proof of \eqref{groumph} using
\eqref{lesbornes}
and the fact that $a(\cdot)\le 1$. The proof of Lemma \ref{inclus} is also concluded.
\end{proof}

\begin{proof}[Proof of Lemma \ref{voldechi}] 

As mentioned in Remark \ref{rem:bij},
there is a natural bijection between $\gamma(t)=\partial\cD_t$ and
$\chi(t)=\partial \cD_t^{(\l(T-t))}$, given by
$\gamma(t)\ni x\mapsto x- \gl(T-t){\bf N}\in \chi(t)$.  It is here
that one needs to have $\l$ smaller than some $\l_0(\cD)$, to guarantee that   $\l (T-t)<m(\cD_t)$. As observed in
\eqref{eq:6}, for points that are in correspondence one has
\begin{equation}
\label{eq:16}
k=\frac{\kappa}{1-\kappa(T-t)\lambda} 
\end{equation}
where we remind that $k$ is the curvature of $\gamma(t)$
and $\kappa$ the curvature of $\chi(t)$.
Moreover Equation \eqref{eq:mc} and the definition of $\chi(t)$ give
\begin{equation}
\label{eq:17}
 \partial_t \chi:=(k+\lambda){\bf N}.
 \end{equation}
The desired equation \eqref{motus} immediately follows from
\eqref{eq:16} and \eqref{eq:17}.
\end{proof}
  
\begin{remark}
  \label{rem:genj} The only properties of $\cD$ we used in the proof
  of the 
  $(j=0)\Rightarrow (j=1)$ step of Proposition \ref{infinit} are $m(\cD)>0$ and the fact that the curvature
  function is bounded and uniformly continuous. Since this continues
  to hold up to time $T(\cD)-c$, the proof for $j>0$ such that
  $j\gep<T(\cD)-c$ works exactly the same, and the small parameters $\gep$
  and $\l$
  can be chosen to be independent of $j$.
\end{remark}


\section*{Acknowledgments}

The authors would like to thank Jimmy Lamboley for enlightening discussions and valuable
comments on various analytical aspects of the present work. 
F. T. was partially supported by ANR project SHEPI: ``Syst\`emes Hors \'Equilibre
de Particules en Interaction''.

\end{document}